\documentclass[lettersize,journal]{IEEEtran}
\usepackage[utf8]{inputenc}
\usepackage{amsfonts}
\usepackage{amsmath}
\usepackage{amsthm}
\usepackage{amssymb}
\usepackage{bm}
\usepackage{color}
\usepackage{graphicx}
\usepackage{float}
\usepackage[capitalise]{cleveref}
\newtheorem{theorem}{Theorem}
\newtheorem{lemma}{Lemma}

\newtheorem{remark}{Remark}

\usepackage{xcolor}

\usepackage{tikz,pgfplots}
\usetikzlibrary{arrows,shapes,backgrounds,patterns,fadings,matrix,arrows,calc,
	intersections,decorations.markings,
	positioning,arrows.meta}
\usepgfplotslibrary{fillbetween}
\usepgfplotslibrary{statistics}
\pgfplotsset{width=5\columnwidth /5, compat = 1.13,
	height = 60\columnwidth /100, grid= major,
	legend cell align = left, ticklabel style = {font=\scriptsize},
	every axis label/.append style={font=\small},
	legend style = {font={\scriptsize}},title style={yshift=-7pt, font = \small} }
\usepackage{verbatim}
\usepackage{makecell}
\usepackage{multirow}
\usepackage{caption}
\usepackage{subcaption}

\begin{document}
	\title{Coevolution of Opinion Dynamics and Recommendation System: Modeling, Analysis and Reinforcement Learning Based Manipulation}
	
	\author{Yuhong Chen$^1$,
		Xiaobing Dai$^{2*}$,
		Martin Buss$^1$,~\IEEEmembership{Fellow,~IEEE
			and Fangzhou Liu$^{3*}$,~\IEEEmembership{Member,~IEEE,}}
		\thanks{
			$^*$Corresponding authors: Xiaobing Dai and Fangzhou Liu.
		}
		\thanks{
			$^1$Chair of Automatic Control Engineering (LSR), School of Computation, Information and Technology (CIT), Technical University of Munich (TUM), 80333 Munich, Germany (email: yuhong.chen, mb@tum.de).
		}
		\thanks{
			$^2$Chair of Information-oriented Control (ITR), School of Computation, Information and Technology (CIT), Technical University of Munich (TUM), 80333 Munich, Germany (email: xiaobing.dai@tum.de).
		}
		\thanks{
			$^3$National Key Laboratory of Modeling and Simulation for Complex Systems, Harbin Institute of Technology, 150001, P.R.China (email: fangzhou.liu@hit.edu.cn).
		}
		\thanks{
			This work is supported in part by the National Natural Science Foundation of China under Grant 62373123, by the Federal Ministry of Research, Technology and Space of Germany in the programme of “Souverän. Digital. Vernetzt.” Joint project 6G-life, project identification number: 16KISK002
		}
	}
	
	\markboth{Journal of \LaTeX\ Class Files,~Vol.~XXX, No.~XXX, November~2024}%
	{Shell \MakeLowercase{\textit{et al.}}: Bare Demo of IEEEtran.cls for IEEE Journals}
	
	\maketitle
	
	\begin{abstract}
		In this work, we develop an analytical framework that integrates opinion dynamics with a recommendation system. By incorporating elements such as collaborative filtering, we provide a precise characterization of how recommendation systems shape interpersonal interactions and influence opinion formation. Moreover, the property of the coevolution of both opinion dynamics and recommendation systems is also shown. 
		Specifically, the convergence of this coevolutionary system is theoretically proved, and the mechanisms behind filter bubble formation are elucidated. 
		Our analysis of the maximum number of opinion clusters shows how recommendation system parameters affect opinion grouping and polarization. Additionally, we incorporate the influence of propagators into our model and propose a reinforcement learning-based solution. The analysis and the propagation solution are demonstrated in simulations using the Yelp data set.
	\end{abstract}
	
	
	\IEEEpeerreviewmaketitle
	
	\section{Introduction}

	The evolution of human opinion is always a hot topic from ancient times to the present, such that the governors can understand social influence, predict social behavior and enhance policy-making. 
    With the emergence of the internet and big data, the evolution of human opinion has undergone significant changes, thereby attracting increasing attention from researchers.

    \subsection{Related Works}
    In the era of traditional communication without the Internet, people primarily formed connections through geographical proximity and familial ties. This relationship is typically stable and seldom changed, such that the evolution of human opinions can be regarded as based on a fixed communication network in past studies. However, with the rise of the internet and social media, the establishment and breaking of the connection has become much easier than in the past, leading to highly dynamic networks \cite{kwak_fragile_2011} which then influences the update of opinions \cite{del_vicario_echo_2016}. Moreover, individuals' opinions also significantly influent the change of connections. Specifically, two individuals with similar opinions are easier to get connected online, while two with opposite opinions may dissolve the existing connection \cite{cinelli_echo_2021}. In this context, the dynamics of online opinions is, in fact, a process of co-evolution between opinions and the network. Since online interaction currently dominates most people’s social lives, studying the co-evolutionary process of opinion and network is of great importance.
	
	The opinion dynamics are widely studied based on a fixed network, where the evolution of individual opinion in the static communication topology is explored. To describe such opinion dynamics, some foundational models \cite{degroot_reaching_1974, french_jr_formal_1956, friedkin_social_1990} provide valuable insights into the mechanisms of opinion propagation and social influence. Considering the dynamic social networks resulting from the possible connection shift between individuals, opinion propagation within time-varying networks is also explored \cite{moreau_stability_2005, munz_consensus_2011, he_reinforcement-learning-based_2023}. However, these works neglect the impact of opinions on the network changes. 
	The Bounded Confidence Model \cite{krause_discrete_2000}, \cite{lorenz_continuous_2007-2} is a co-evolutionary framework, in which agents ignore opinions outside their confidence interval. Consequently, the interaction graph is determined by pairwise distance metrics in opinion space and co-evolves with the opinions of all agents.
	This implies that the graph depends specifically on the mutual distances between individuals, co-evolving with their opinions. Other notable models exploring co-evolving opinion dynamics include \cite{gu_co-evolution_2017, kang_coevolution_2022, liu_emergence_2023, liu_modeling_2024}.
	The coevolutionary models contribute a deep and detailed understanding of opinion change by considering the adaptability of social connections.
	
	Despite such advancements, most coevolutionary models still fall short of capturing the complexity of modern online social networks, where the recommendation system plays an important role \cite{berjani2011recommendation}.
	Recommendation systems, mediating a large share of current online interactions, are widely used in online platforms like social media, e-commerce and streaming services. 
	By using algorithms for suggesting content, products, or connections based on user opinions, recommendation systems shape the connection between individuals and influence the structure of networks \cite{isinkaye2015recommendation}. 
	While they facilitate the connection formation process for individuals, the recommendation systems may induce ``filter bubbles'' and ``echo chambers''.
	Specifically, users are more frequently exposed to others who align with their preexisting opinions, reinforcing opinions and limiting exposure to diverse perspectives. 
	The influence of recommendation systems on both the formation of new connections and the reinforcement of existing ones adds complexity to opinion dynamics studies.
	This makes the traditional coevolutionary models unsuitable for opinion evolution within modern online social networks with recommendation systems.
	
	Currently, only limited studies specifically explore the impact of recommendation systems on public opinion. 
	For instance, \cite{bellina_effect_2023, piao_humanai_2023, vendeville_opening_2023} investigate the indirect impact of the recommended items on the change of opinions, but they did not consider the formation and evolution of group opinions from the perspective of opinion dynamics.
	Some have recently attempted to explore opinion dynamics in the context of recommendation systems \cite{cinus_effect_2022, perra_modelling_2019, pescetelli_indirect_2022, ramaciotti_morales_auditing_2021, santos_link_2021, tornberg_how_2022}. 
	However, these works only provide empirical results via simulations and lack a comprehensive analysis of the underlying theoretical principle.  
	The theoretical analysis of opinion dynamics in systems with recommendation mechanisms is conducted in \cite{rossi_closed_2022}, where the authors investigate a tractable analytical model of users interacting with an online news aggregator, aiming to elucidate the feedback loop between the evolution of user opinions and personalized content recommendations. Similarly, \cite{iannelli_filter_2022} explores a multi-state voter model, where users interact with “personalized information” that reflects their historically dominant opinions, with a given probability. 
	However, in both studies \cite{rossi_closed_2022, iannelli_filter_2022}, the recommendation system only provides content selection methods for e.g., news items, and neglects its effect on user–user connections.
	Therefore, the analysis of the co-evolution of interpersonal networks is excluded from their works.
	Furthermore, substantial research focuses on modeling and predicting opinion dynamics within systems with recommendation algorithms, overlooking the potential influence exerted by propagators to steer public opinion.

    \subsection{Contributions}
	In this work, we propose an analytical framework combined with opinion dynamics and a recommendation system to study their coevolution. We provide a precise characterization of how recommendation systems influence interpersonal networks, accommodating complex behaviors such as collaborative filtering. The detailed contribution is summarized as follows.
	\begin{itemize}
		\item Theoretically, we prove the convergence of this complex system, offering an explanation for the formation of filter bubbles.
		\item By establishing a relationship between the upper bound on the number of opinion clusters and the solution to the sphere-packing problem, we demonstrate how recommendation system parameters impact opinion grouping and polarization. 
		\item We incorporate the role of propagators into the proposed framework, to shape group opinions. Specifically, the propagation behavior is regarded as an optimization problem, which minimizes the influence cost while achieving the formation of the predefined group opinions.
		\item We address this optimal control problem for propagators within the complex co-evolving system by employing a reinforcement-learning-based controller.
	\end{itemize}
	Furthermore, the effectiveness of the proposed framework and the controller is demonstrated via simulations, where the feasibility of controlling opinion propagation within the framework is shown.
	
	The remaining content of this paper is structured as follows.
	\cref{section_preliminary} presents the foundation of the opinion dynamics and recommendation system. In \cref{section_analysis}, the co-evolutionary framework combining opinion dynamics with recommendation systems is proposed, and its performance is analyzed.
	The effect of propagators and their reinforcement learning-based algorithm are introduced in \cref{section_control}.
	\cref{section_experiment} encompasses the execution of numerical simulations, followed by the conclusion of this paper in \cref{section_conclusion}.
	
	\section{Preliminaries and Problem Setting}
	\label{section_preliminary}
	
	In this section, the fundamental knowledge of opinion dynamics and recommendation systems is introduced in \cref{subsection_opinion_dynamics} and \cref{subsection_recommendation_system} respectively, inducing the problem setting of this paper.
	
	\subsection{Multidimensional Opinion Dynamics}
	\label{subsection_opinion_dynamics}
	
	In this work, we consider a system including $n \in \mathbb{N}_+$ users, and each user $i$ has a multi-dimensional individual opinion $\bm{x}_i \in [0,1]^m$ to $m \in \mathbb{N}_+$ items, i.e., $\bm{x}_i = [x_i^1, \cdots, x_i^m]^T$ for $\forall i = 1, \cdots, m$.
	The opinion of each user $i$ at step $k \in \mathbb{N}$ is affected by the opinions from other users at the same time, forming the opinion dynamics as
	\begin{align} \label{eqn_single_opinion_dynamics}
		\bm{x}_i(k + 1) = \sum\nolimits_{j = 1}^n w_{i,j}(k) \bm{x}_j(k), &&\forall i=1,\cdots,n,
	\end{align}
	where the weights $w_{i,j}(k) \in [0,1] \subset \mathbb{R}$ denote the influence strength from user $j$ to $i$ at step $k$.
	Moreover, $w_{i,j}(k)$ also reflects the connectivity between each user, which is also defined as a time-varying graph $\mathcal{G}(k) = \{ \mathcal{V}, \bm{W}(k) \}$ with vertex set $\mathcal{V} = \{ 1, \cdots, n \}$ and adjacent matrix $\bm{W}(k) = [w_{i,j}(k)]_{i,j = 1,\cdots,n}\in [0,1]^{n \times n} \subset \mathbb{R}^{n \times n}$.
	Define the concatenation of $\bm{x}_i$ as $\bm{X} = [\bm{x}_1^T, \cdots, \bm{x}_n^T] \in [0,1]^{n \times m}$, then the overall multi-dimensional opinion dynamics is written as
	\begin{align} \label{eqn_Degr}
		\bm{X}(k+1) = \bm{W}(k) \bm{X}(k),
	\end{align}
	following the time-varying French-DeGroot model.
	Note that \eqref{eqn_Degr} illustrates how opinions $\bm{X}(k)$ are influenced by the network $\bm{W}(k)$.
	The topology reflected by $\bm{W}(k)$ is affected by the recommendation system, particularly in modern communication over the Internet, which is shown in the next subsection.
	
	\subsection{Recommendation System with Different Metrics}
	\label{subsection_recommendation_system}
	
	Recommendation systems aim to build the connection between two users $i$ and $j$ with $i, j \in \mathcal{V}$ according to the similarity of individual opinions $\bm{x}_i$ and $\bm{x}_j$. 
	To measure the similarity of two opinions $\bm{x}_i$ and $\bm{x}_j$, several methods are employed, such as Euclidean- and angle-based methods.
	Specifically, the Euclidean distance between $\bm{x}_i$ and $\bm{x}_j$ denotes
	\begin{align}\label{dis}
		\mathrm{dis}(\bm{x}_i, \bm{x}_j) = \| \bm{x}_i - \bm{x}_j \|, && \forall i, j \in \mathcal{V},
	\end{align}
	which is composed of the absolute value of $x_i^p - x_j^p$ for every dimension $p = 1, \cdots, m$.
	However, the Euclidean-based measurement overlooks the importance of the ratio between two dimensions $x_i^p$ and $x_i^q$ for any user $i \in \mathcal{V}$ and $p,q = 1,\cdots,m$.
	This induces potentially low similarity between users with similar opinion combinations but different strengths, i.e., for users $i$ and $j$ with $\bm{x}_i = \xi \bm{x}_j$, $\xi \in \mathbb{R}_{0,+}$.
	Taking the ratio $x_i^p / x_i^q$ for any $p,q = 1,\cdots,m$ into consideration, the angle-based cosine similarity is used, which is defined as
	\begin{align} \label{cs}
		\mathrm{sim}(\bm{x}_i,\bm{x}_j) = \frac{ \bm{x}_i^T \bm{x}_j }{\| \bm{x}_i \| \| \bm{x}_j \|},
	\end{align}
	reflecting the cosine of the angle between $\bm{x}_i$ and $\bm{x}_j$.
	While the cosine returns the value in $[-1,1]$, the domain of $\bm{x}_i \in [0,1]^m$ ensures the output of $\mathrm{sim}(\cdot,\cdot)$ is always non-negative.
	
	\begin{remark}
		Note that $\mathrm{dis}(\cdot, \cdot)$ and $\mathrm{sim}(\cdot, \cdot)$ reflect the similarity of users' opinions in a totally opposite way.
		Specifically, $\mathrm{dis}(\cdot, \cdot)$ results in a small value for users with close opinions, while $\mathrm{sim}(\cdot, \cdot)$ returns a large value.
	\end{remark}
	
	In this section, we introduce the fundamental knowledge of opinion dynamics and recommendation systems. Specifically, we aim to develop a comprehensive framework that integrates opinion dynamics with recommendation systems, enabling an analytical examination of the co-evolution between individual opinions and the underlying communication network.
	
	\section{Framework for Opinion Dynamics with Recommendation System}
	\label{section_analysis}
	
	In this section, we propose a framework combining the opinion dynamics with recommendation systems (ODRS) in \cref{subsection_ODRS}.
	Moreover, the convergence and clustering property of the opinion evolution within the proposed framework is analyzed in \cref{subsection_convergency} and \cref{subsection_clustering}, respectively.
	
	\subsection{ODRS Framework}
	\label{subsection_ODRS}
	
	Consider the weights $w_{ij}(\cdot)$ in \eqref{eqn_single_opinion_dynamics} reflect the influence strength from user $j$ to $i$ with $i, j \in \mathcal{V}$, which is determined by recommendation systems within the modern online communication environment, as shown in \cref{subsection_recommendation_system}.
	To represent the influence strength, the degree of connectivity function $s(\cdot, \cdot): [0,1]^m \times [0,1]^m \to \mathbb{R}_{0,+}$ is introduced for both distance- and angle-based methods in \cref{subsection_recommendation_system} with $\mathrm{dis}(\cdot, \cdot)$ and $\mathrm{sim}(\cdot, \cdot)$, respectively.
	Moreover, considering that the domain of $w_{ij}(\cdot)$ in \eqref{eqn_single_opinion_dynamics} is $[0,1]$, the value of $w_{ij}(\cdot)$ is determined by the normalization as
	\begin{equation} \label{eqn_w}
		w_{ij}(k) = \frac{s (\bm{x}_i(k), \bm{x}_j(k))}{\sum_{j \in \mathcal{V}} s (\bm{x}_i(k), \bm{x}_j(k))}
	\end{equation}
	for any $k \in \mathbb{N}$.
	Then, the opinion dynamics in \eqref{eqn_single_opinion_dynamics} becomes
	\begin{equation} \label{OD_S}
		\bm{x}_i(k + 1) = \frac{1}{\sum_{j \in \mathcal{V}} s (\bm{x}_i(k), \bm{x}_j(k))} \sum_{j \in \mathcal{V}} s (\bm{x}_i(k), \bm{x}_j(k)) \bm{x}_j(k),
	\end{equation}
	whose concatenation form with $\bm{X}$ in \eqref{eqn_Degr} is written as
	\begin{align} \label{OD_S_concatenated}
		\bm{X}(k + 1) = (\mathrm{diag}(\bm{S}(k) \bm{1}_n))^{-1} \bm{S}(k) \bm{X}(k)
	\end{align}
	with $S(k) = [s(\bm{x}_i(k), \bm{x}_j(k))]_{i,j \in \mathcal{V}}$.
	The dynamics in \eqref{OD_S} shows the influence of the recommendation system in the structure of opinion dynamics, which is strongly dependent on the value of $s (\cdot, \cdot)$.
	Next, the formulation of $s (\cdot, \cdot)$ is determined under the distance- and angle-based methods.
	
	For the distance-based method, the larger the distance between two opinions is, the less similar the users are.
	Therefore, the function $s(\cdot, \cdot)$ for distance-based method is defined as
	\begin{align} \label{s_f}
		s (\bm{x}_i, \bm{x}_j) = f(\mathrm{dis}(\bm{x}_i, \bm{x}_j)), && \forall i, j \in \mathcal{V}, 
	\end{align}
	where $f(\cdot): \mathbb{R}_{0,+} \to \mathbb{R}_{0,+}$ is a non-increasing function.
	In contrast, for the angle-based method, it holds that the larger the cosine similarity is, the smaller the angle between the users' opinions is.
	This also indicates the more similar the users are, such that the degree of connectivity $s(\cdot, \cdot)$ should be monotonically increasing with respect to cosine similarity $\mathrm{sim}(\cdot, \cdot)$.
	Formally, the degree of connectivity function $s(\cdot, \cdot)$ is defined as
	\begin{align} \label{s_g}
		s (\bm{x}_i, \bm{x}_j) = g(\mathrm{sim}(\bm{x}_i, \bm{x}_j)), && \forall i, j \in \mathcal{V}, 
	\end{align}
	with $g: [0,1] \to \mathbb{R}_{0,+}$ as a non-decreasing function.
	It is obvious to see that larger $s (\bm{x}_i, \bm{x}_j)$ indicates a larger effect from user $j$ to user $i$.
	Moreover, $s (\bm{x}_i, \bm{x}_j) = 0$ means the opinion $\bm{x}_j$ will not be recommended to user $i$. 
	
	Furthermore, considering the fact that not every user is connected to the others, especially with very low similarity in individual opinion, the function $s (\bm{x}_i, \bm{x}_j)$ is not continuous but truncated.
	Specifically, for any user $i \in \mathcal{V}$, the recommendation system usually shows only a few other users, who have sufficiently similar opinions as $i$, instead of exploring everyone.
	This behavior is formulated for the distance-based method as
	\begin{align} \label{f_bar}
		f(\mathrm{dis}(\bm{x}_i, \bm{x}_j)) = \begin{cases}
			1 &, \text{if}~ \mathrm{dis}(\bm{x}_i,\bm{x}_j) \le \sqrt{m} (1 - \epsilon) \\
			0 &, \text{otherwise}
		\end{cases},
	\end{align}
	where $\epsilon \in [0,1] \subset \mathbb{R}$ is a predefined threshold.
	Considering the state domain of $\bm{x}_i$ as $[0,1]^m$ for any $i \in \mathcal{V}$ inducing
	\begin{align}
		\mathrm{dis}(\bm{x}_i,\bm{x}_j) = \sqrt{\sum_{p=1}^m (x_i^p - x_j^p)^2} \le \sqrt{\sum_{p=1}^m 1} = \sqrt{m},
	\end{align}
	it is easy to see $f(\mathrm{dis}(\bm{x}_i, \bm{x}_j)) = \mathrm{dis}(\bm{x}_i,\bm{x}_j)$ for $\epsilon = 0$.
	Moreover, for $\epsilon = 1$, only the users with identical opinions affect each other.
	For angle-based method using cosine similarity, the function $g(\cdot)$ is chosen as
	\begin{align} \label{g_bar}
		g(\mathrm{sim}(\bm{x}_i, \bm{x}_j)) \!=\! \begin{cases}
			\mathrm{sim}(\bm{x}_i, \bm{x}_j) &, \text{if}~ \mathrm{sim}(\bm{x}_i,\bm{x}_j) \!\ge\! \epsilon \\
			0 &, \text{otherwise}
		\end{cases},
	\end{align}
	where $\epsilon = 0$ induces $g(\mathrm{sim}(\bm{x}_i, \bm{x}_j)) = \mathrm{sim}(\bm{x}_i, \bm{x}_j)$.
	For $\epsilon = 1$, the nonzero $g(\mathrm{sim}(\bm{x}_i, \bm{x}_j))$ indicates $\bm{x}_i = \xi \bm{x}_j$ with $\xi \in \mathbb{R}_+$.
	
	In summary, the ODRS framework is shown as in \cref{fig_structure}, whose performance is analyzed in the following subsections.
	\begin{figure}[t]
		\centering
		\includegraphics[width = 0.48\textwidth]{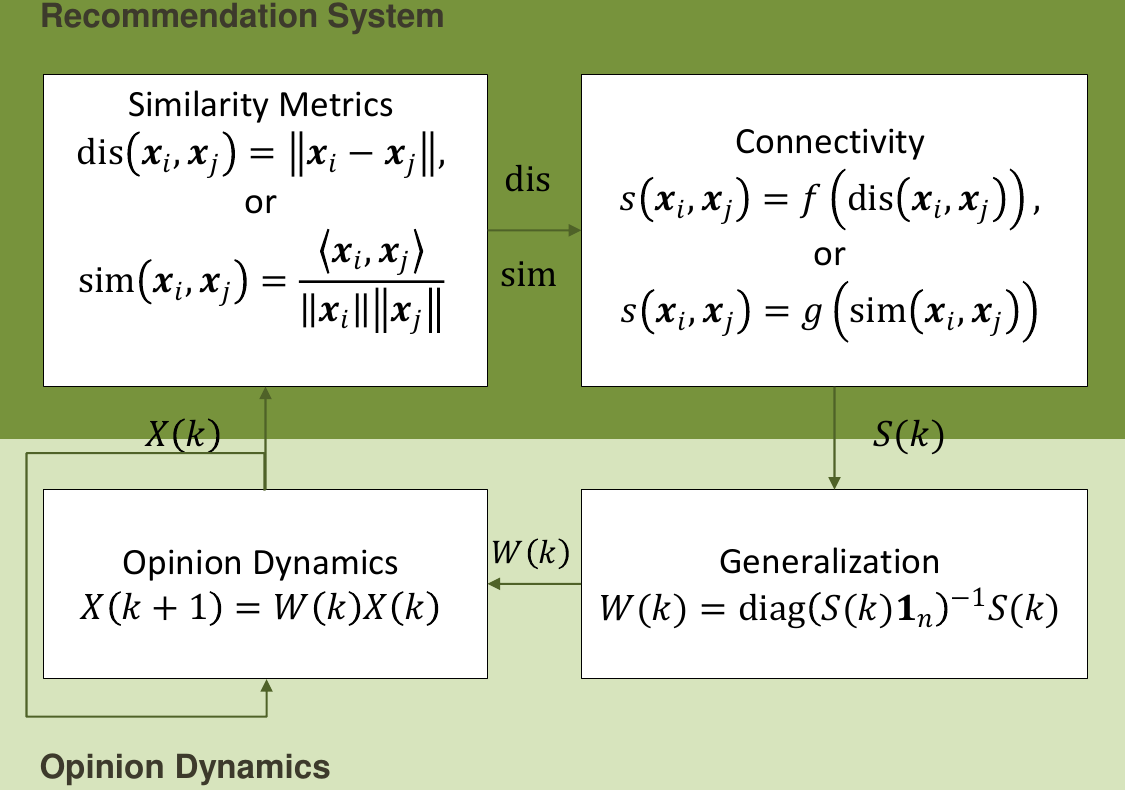}
		\caption{Structure of opinion dynamics with recommendation system}
		\label{fig_structure}
	\end{figure}
	
	\subsection{Convergency Analysis}
	\label{subsection_convergency}
	
	In this subsection, the convergence property of the proposed ODRS framework is analyzed using both distance- and angle-based methods.
	Before the formal analysis, some intermediate results are shown as follows.
	
	\begin{lemma} \label{lemma_range_x}
		Consider the ODRS system \eqref{OD_S_concatenated} with $w_{ij}$ defined in \eqref{eqn_w} and initial state as $\bm{x}_i(0) \in [0,1]^m$ for all $i \in \mathcal{V}$, then it has $\bm{x}_i(k)  \in [0,1]^m$ for any $k \in \mathbb{N}$.
	\end{lemma}
	\begin{proof}
		Considering the form of $w_{ij}(k)$ in \eqref{eqn_w}, the weight matrix $\mathbf{W}(k)$ is a row-stochastic and non-negative matrix for any $k \in \mathbb{N}$, such that
		\begin{align}
			x_i^p(k+1) =\sum\limits_{j=1}^{n} w_{ij}(k) x_j^p(k) \leq \sum\limits_{j=1}^{n} w_{ij}(k) = 1
		\end{align}
		holds for any $p = 1, \cdots, m$ and $i \in \mathcal{V}$.
		Similarly, it has
		\begin{align}
			x_i^p(k+1) &=\sum\limits_{j=1}^{n} w_{ij}(k) x_j^p(k) \geq \sum\limits_{k=1}^{m} 0 = 0,
		\end{align}
		resulting in $0 \leq  x_i^p(k+1)  \leq 1$ and $\bm{x}_i(k + 1) \in [0,1]^m$, which concludes the proof.
	\end{proof}
	
	\begin{remark}
		\cref{lemma_range_x} indicates that the state domain $[0,1]^m$ is forward invariant within the proposed ODRS framework.
		It is worth noting that the result in Lemma 1 relies only on the initial value $\bm{x}_i(0)$ and normalized weights $w_{ij}$ as defined in \eqref{eqn_w}. Consequently, Lemma 1 remains valid in the convergence analysis of both distance-based and angle-based methods.
	\end{remark}
	
	With the result in \cref{lemma_range_x}, the convergence of the ODRS with distance- and angle-based methods is investigated below.
	
	\subsubsection{Convergence of ODRS with distance-based method}
	
	We first analyze the convergence of the proposed ODRS framework with the distance-based method using \eqref{s_f} and \eqref{f_bar}, which is shown as follows.
	
	\begin{theorem}
		Consider the co-evolutionary opinion dynamics with the recommendation system in \eqref{OD_S} using the distance-based method with \eqref{s_f}, and \eqref{f_bar}. The coevolution system terminates.
	\end{theorem}
	\begin{proof}
		See \cite{nedic_multidimensional_2012}.
	\end{proof}
	
	In fact, with the threshold $\epsilon$ shown in \eqref{f_bar}, the system effectively becomes a multidimensional Hegselmann-Krause (HK) model, where a polynomial upper bound on the convergence time is shown as in \cite{bhattacharyya_convergence_2013,etesami_game-theoretic_2015-1}.
	Furthermore, convergence is proven for such multidimensional HK models with a broad class of norms to measure distances \cite{proskurnikov_recurrent_2020}. More results can be found in \cite{bernardo_bounded_2024}, which highlight the extensibility of the proposed ODRS framework with the distance-based method.
	
	\subsubsection{Convergence of ODRS with angle-based method}
	
	Next, the convergence of ODRS with the angle-based method is investigated, for which we define the diameter of the opinion matrix $\bm{X}$, i.e., $d(\bm{X})$.
	Specifically, the diameter $d(\bm{X})$ is defined as the largest angle between two opinions in $\bm{X}$, which is written as
	\begin{equation}
		d(\bm{X}) = \max\nolimits_{i,j \in \mathcal{V}} \{ \arccos (\mathrm{sim}(\bm{x}_i, \bm{x}_j)) \}
	\end{equation}
	using the cosine similarity defined in \eqref{cs}.
	It is easy to see that $d(\bm{X})$ reflects the maximal difference of the opinions for users $\mathcal{V}$, whose evolution under ODRS with angle-based method \eqref{s_g} and \eqref{g_bar} is shown as follows.
	
	\begin{theorem} \label{theorem_angle_convergence}
		Consider the co-evolutionary opinion dynamics with the recommendation system in \eqref{OD_S} using the angle-based method with \eqref{s_g}, and \eqref{g_bar}.
		The diameter $d(\bm{X})$ of the opinion set $\bm{X}$ satisfies
		\begin{align}
			d(\bm{X}(k + 1)) \le d(\bm{X}(k)), && \forall k \in \mathbb{N}
		\end{align}
		for any initial state $\bm{X}(0) \in [0,1]^{m \times n}$.
	\end{theorem}
	\begin{proof}
		We consider the opinions $\bm{x}_i(k + 1), \bm{x}_j(k + 1)$ for any two users $i, j \in \mathcal{V}$ at step $k + 1$, which are written as
		\begin{align}
			&\bm{x}_i(k + 1) = \sum\nolimits_{r = 1}^n w_{ir}(k) \bm{x}_r(k), \\
			&\bm{x}_j(k + 1) = \sum\nolimits_{s = 1}^n w_{js}(k) \bm{x}_s(k)
		\end{align}
		with $k \in \mathbb{N}$ following ODRS framework in \eqref{OD_S}.
		Then, the cosine similarity between $\bm{x}_i(k + 1)$ and $\bm{x}_j(k + 1)$ defined in \eqref{cs} is then written as
		\begin{align}
			\mathrm{sim}&(\bm{x}_i(k + 1), \bm{x}_j(k + 1)) \nonumber \\
			&= \frac{(\sum_{r = 1}^n w_{ir}(k) \bm{x}_r(k))^T (\sum_{s = 1}^n w_{js}(k) \bm{x}_s(k))}{\| \sum_{r = 1}^n w_{ir}(k) \bm{x}_r(k) \| \| \sum_{s = 1}^n w_{js}(k) \bm{x}_s(k) \|} \\
			&= \frac{\sum_{r = 1}^n \sum_{s = 1}^n w_{ir}(k) w_{js}(k) \bm{x}_r^T(k) \bm{x}_s(k)}{\| \sum_{r = 1}^n w_{ir}(k) \bm{x}_r(k) \| \| \sum_{s = 1}^n w_{js}(k) \bm{x}_s(k) \|}. \nonumber
		\end{align}
		Considering the result in \cref{lemma_range_x} with the initial condition $\bm{X}(0) \in [0,1]^{m \times n}$, it is easy to see $\bm{x}_i(k) \in [0,1]^m$ for any $i \in \mathcal{V}$ and $k \in \mathbb{N}$, such that
		\begin{align}
			\bm{x}_r^T(k) \bm{x}_s(k) = \sum\nolimits_{p = 1}^m x_r^p x_s^p \ge \sum\nolimits_{p = 1}^m 0 = 0
		\end{align}
		for any $r, s \in \mathcal{V}$.
		Therefore, the value of $\mathrm{sim}(\bm{x}_i(k + 1), \bm{x}_j(k + 1))$ is lower bounded by
		\begin{align} \label{eqn_sim_next_2}
			\mathrm{sim}&(\bm{x}_i(k + 1), \bm{x}_j(k + 1)) \nonumber \\
			&\ge \frac{\sum_{r = 1}^n \sum_{s = 1}^n w_{ir}(k) w_{js}(k) \bm{x}_r^T(k) \bm{x}_s(k)}{(\sum_{r = 1}^n w_{ir}(k) \| \bm{x}_r(k) \|)  (\sum_{s = 1}^n w_{js}(k) \| \bm{x}_s(k) \|)} \\
			&= \frac{\sum_{r = 1}^n \sum_{s = 1}^n w_{ir}(k) w_{js}(k) \bm{x}_r^T(k) \bm{x}_s(k)}{\sum_{r = 1}^n \sum_{s = 1}^n w_{ir}(k) w_{js}(k) \| \bm{x}_r(k) \| \| \bm{x}_s(k) \|} \nonumber
		\end{align}
		using the triangle inequality.
		Moreover, define $\underline{i}(k), \underline{j}(k) \in \mathcal{V}$ as the user indices satisfying
		\begin{align}
			\mathrm{sim}(\bm{x}_{\underline{i}(k)}(k), \bm{x}_{\underline{j}(k)}(k)) \!\!=\!\! \min\nolimits_{i,j \!\in\! \mathcal{V}} \{ \mathrm{sim}(\bm{x}_i(k), \bm{x}_j(k)) \},
		\end{align}
		then it is obvious to see
		\begin{align} \label{eqn_lowerbound_xrxs}
			\bm{x}_r^T(k) \bm{x}_s(k) \!\!\ge\!\! \frac{\bm{x}_{\underline{i}(k)}^T(k) \bm{x}_{\underline{j}(k)}(k)}{\| \bm{x}_{\underline{i}(k)}(k) \| \| \bm{x}_{\underline{j}(k)}(k) \|} \| \bm{x}_r(k) \| \| \bm{x}_s(k) \|.
		\end{align}
		Apply \eqref{eqn_lowerbound_xrxs} into \eqref{eqn_sim_next_2}, the cosine similarity $\mathrm{sim}(\bm{x}_i(k + 1), \bm{x}_j(k + 1))$ is further bounded by
		\begin{align}
			\mathrm{sim}(\bm{x}_i&(k + 1), \bm{x}_j(k + 1)) \nonumber \\
			\ge& \frac{\sum_{r = 1}^n \sum_{s = 1}^n w_{ir}(k) w_{js}(k) \| \bm{x}_r(k) \| \| \bm{x}_s(k) \|}{\sum_{r = 1}^n \sum_{s = 1}^n w_{ir}(k) w_{js}(k) \| \bm{x}_r(k) \| \| \bm{x}_s(k) \|} \\
			&\times \frac{\bm{x}_{\underline{i}(k)}^T(k) \bm{x}_{\underline{j}(k)}(k)}{\| \bm{x}_{\underline{i}(k)}(k) \| \| \bm{x}_{\underline{j}(k)}(k) \|} \nonumber \\
			=& \min\nolimits_{i,j \in \mathcal{V}} \{ \mathrm{sim}(\bm{x}_i(k), \bm{x}_j(k)) \}. \nonumber
		\end{align}
		Furthermore, considering the definition of $d(\cdot)$ with $\arccos(\cdot)$ as a monotonically decreasing function, it has
		\begin{align}
			\arccos&(\mathrm{sim}(\bm{x}_i(k + 1), \bm{x}_j(k + 1))) \nonumber \\
			&\le \arccos(\min\nolimits_{i,j \in \mathcal{V}} \{ \mathrm{sim}(\bm{x}_i(k), \bm{x}_j(k)) \}) \\
			&= \max\nolimits_{i,j \in \mathcal{V}} \{ \arccos(\mathrm{sim}(\bm{x}_i(k), \bm{x}_j(k))) \} = d(\bm{X}(k)) \nonumber
		\end{align}
		for any $i, j \in \mathcal{V}$ and $k \in \mathbb{N}$, resulting in $d(\bm{X}(k + 1)) = \max_{i,j \in \mathcal{V}} \arccos(\mathrm{sim}(\bm{x}_i(k + 1), \bm{x}_j(k + 1))) \le d(\bm{X}(k))$.
		This concludes the proof.
	\end{proof}
	
	\cref{theorem_angle_convergence} shows the convergence of the diameter $d(\cdot)$ over time, indicating that the maximal difference of users' opinions will decrease.
	However, \cref{theorem_angle_convergence} implies no asymptotic convergence of the users' opinion on one specific value.
	Instead, the ultimate boundness of $d(\cdot)$ is ensured \cite{khalil2015nonlinear} by a well-define ultimate bound $\bar{d} \in \mathbb{R}_{0,+}$, and it allows the opinion divergence satisfying $d(\bm{X}(\infty)) \le \bar{d}$.
	Moreover, some clustering results can be obtained within the proposed ODRS framework, which is shown in the next subsection.
	
	\subsection{Clustering Analysis}
	\label{subsection_clustering}
	
	While the convergence is proved in the previous subsection inducing ultimate boundness of the opinion difference, in this subsection more detailed discussion is conducted for the behavior inner the ultimate bound.
	A typical steady-state behavior exhibited by bounded confidence opinion dynamics is the coincidence of the opinions in each disconnected subgroup of agents. This situation is called clustering \cite{bernardo_bounded_2024}. A cluster is a completely isolated component of a graph.
	
	\begin{lemma} \label{lemma_cluster}
		For the co-evolutionary ODRS in \eqref{OD_S}, using either the angle-based method with \eqref{s_g} and \eqref{g_bar} or the distance-based method with \eqref{s_f} and \eqref{f_bar}, the system achieves either consensus or clustering.
	\end{lemma}
	\begin{proof}
		See \cite{proskurnikov_tutorial_2018} Lemma 1, and Theorem 13.
	\end{proof}
	
	\cref{lemma_cluster} shows that the individual opinions following the ODRS framework achieve clustering, i.e., separating the opinions into several groups.
	Intuitively, clustering is formed due to the potential disconnection between users, which is determined by the parameter $\epsilon$ in \eqref{f_bar} and \eqref{g_bar} for distance- and angle-based methods.
	Therefore, it is intended to analyze the influence of $\epsilon$ on the clustering behavior reflected by the number of clustered groups, which is shown as follows.
	
	\subsubsection{Cluster Numbers with distance-based method}
	
	For the distance-based method with \eqref{s_f} and \eqref{f_bar}, the convergence case is discussed.
	Define two different clustered user groups without connection as $\mathcal{X}_i$ and $\mathcal{X}_j$ with $i, j \in \mathbb{N}$, whose converged opinion denote $\bar{\bm{x}}_i$ and $\bar{\bm{x}}_j$ respectively.
	Note that it holds $\bar{\bm{x}}_i \in [0,1]^m$ for $\forall i \in \mathbb{N}$ due to \cref{lemma_range_x}.
	From the formation in \eqref{f_bar}, it is easy to see the disconnection between $\mathcal{X}_i$ and $\mathcal{X}_j$ indicates $\mathrm{dis}(\bar{\bm{x}}_i, \bar{\bm{x}}_j) > \sqrt{m}(1 - \epsilon)$ for any $i, j \in \mathbb{N}$.
	Then, the number of clustered opinion groups equals to the cardinality of $\{ \bar{\bm{x}}_i \}_{i \in \mathbb{N}}$, such that $\mathrm{dis}(\bar{\bm{x}}_i, \bar{\bm{x}}_j) > \sqrt{m}(1 - \epsilon)$ for any $i, j \in \mathbb{N}$.
	While the exact number of clustered opinion groups is hard to obtain, its upper bound is relevant to the covering number, which is shown as follows.
	
	\begin{theorem} \label{theorem_cluster_distance}
		Consider the convergence state for the co-evolutionary opinion dynamics with recommendation system and $n$ users in \eqref{OD_S} using the distance-based method with \eqref{s_f}, and \eqref{f_bar} with predefined $\epsilon \in [0,1]$.
		The number of the clustered opinion groups is upper bounded by $\min\{ N_f(1 - \epsilon), n \}$ with $N_f(\tau) = \lfloor 1 / \tau + 1 \rfloor$.
	\end{theorem}
	\begin{proof}
		The proof employs the concept of covering numbers.
		Specifically, choose the grid factor as $\delta \in (0,1) \subset \mathbb{R}$, and define the maximal number of division for $[0,1]$ along each dimension as $\bar{i} = \lfloor 1 / \tau + 1 \rfloor$.
		Moreover, define the $m$-dimensional set as $\mathbb{X}(\bm{i}) = [(i_1 - 1) / \bar{i}, i_1 / \bar{i}) \times \cdots \times [(i_m - 1) / \bar{i}, i_m / \bar{i})$, where $\bm{i} = [i_1, \cdots, i_m]^T \in \{ 1, \cdots, \bar{i} \}^m \subset \mathbb{N}^m$.
		Note that the sets $\mathbb{X}(\bm{i})$ satisfy
		\begin{align}
			\mathbb{X}(\bm{i}) \cap \mathbb{X}(\bm{i}') = \emptyset, \qquad \forall \bm{i} \ne \bm{i}' \in \{ 1, \cdots, \bar{i} \}^m.
		\end{align}
		For any converged opinion $\bar{\bm{x}}_r = [\bar{x}_r^1, \cdots, \bar{x}_r^m]^T$ and $\bar{\bm{x}}_s = [\bar{x}_s^1, \cdots, \bar{x}_s^m]^T$ with $r \ne s \in \mathbb{N}$ and assuming $\bar{\bm{x}}_r, \bar{\bm{x}}_s \in \mathbb{X}(\bm{i})$ for a specific $\bm{i} \in \{ 1, \cdots, \bar{i} \}^m$, the Euclidean distance between $\bar{\bm{x}}_r$ and $\bar{\bm{x}}_s$ is bounded by
		\begin{align}
			\| \bar{\bm{x}}_r - \bar{\bm{x}}_s \|^2 = \sum_{j = 1}^m (\bar{x}_r^j - \bar{x}_s^j)^2 \le \sum_{j = 1}^m \tau^2 = m \tau^2,
		\end{align}
		resulting in $\| \bar{\bm{x}}_r - \bar{\bm{x}}_s \| \le \sqrt{m} \tau$.
		Let $\tau = 1 - \epsilon$, and it is obvious that $\bar{\bm{x}}_r$ and $\bar{\bm{x}}_s$ belong to different set $\mathbb{X}(\cdot)$, i.e., $\bar{\bm{x}}_r \in \mathbb{X}(\bm{i}_r)$ and $\bar{\bm{x}}_s \in \mathbb{X}(\bm{i}_s)$ with $\bm{i}_r \ne \bm{i}_s \in \{ 1, \cdots, \bar{i} \}^m$.
		Moreover, assume $\bar{i}^m$ converged opinions exist in $[0,1]^m$ and consider
		\begin{align}
			[0,1]^m \subseteq \bigcup_{i_1 = 1}^{\bar{i}} \cdots \bigcup_{i_m = 1}^{\bar{i}} \mathbb{X}([i_1, \cdots, i_m]^T),
		\end{align}
		then any additional converged opinion $\bar{\bm{x}}_{i^*}$ must share one of the set $\mathbb{X}(\cdot)$ with one other converged opinion, denoted as $\bar{\bm{x}}_{j^*}$ with $i^*, j^* \in \mathbb{N}$.
		This induces the connection between $\bar{\bm{x}}_{i^*}$ and $\bar{\bm{x}}_{j^*}$, which means $\bar{\bm{x}}_{i^*}$ is not the converged opinion of a new cluster.
		This also indicated $\bar{i}^m$ is the maximal number of converged opinions corresponding to the threshold $\epsilon$, which concludes the proof by additionally considering the clustering number is smaller than the user number.
	\end{proof}
	
	\cref{theorem_cluster_distance} shows the upper bound of the number of potential clustered opinion groups, which is monotonically increasing with respect to $\epsilon$.
	This observation is intuitive since a larger $\epsilon$ induces a high probability of losing the connection between two users by considering the definition of $f(\cdot)$ in \eqref{f_bar}.
	Note that the upper bound is relevant to the covering number $N_f(\cdot)$ with the grid factor $1 - \epsilon$, which can also be calculated by using different ways for different conservatism \cite{lederer2019uniform}.

	\subsubsection{Cluster Numbers with angle-based method}
	
	Similarly, as in the discussion for the distance-based method, the upper bound of the number of clustered opinion groups is investigated for the angle-based method.
	Specifically, considering \eqref{s_g} and \eqref{g_bar}, it is easy to see that two clustered groups $\mathcal{X}_i$ and $\mathcal{X}_j$ with converged opinions $\bar{\bm{x}}_i$ and $\bar{\bm{x}}_j$ satisfy $\mathrm{sim}(\bar{\bm{x}}_i, \bar{\bm{x}}_j) < \epsilon$ for any $i, j \in \mathbb{N}$.
	This also means the angle between $\bar{\bm{x}}_i$ and $\bar{\bm{x}}_j$ is greater than $\arccos(\epsilon)$.
	Therefore, obtaining the upper bound for opinion clusters is identical to finding the maximum number of vectors in a sphere, where the angle between any two vectors is lower bounded by $\arccos(\epsilon)$.
	Such an upper bound is related to the solution for the spherical codes problem, which is generally described as ``Place $N_s$ points on a sphere in $m$ dimensions to maximize the
	minimum angle $\theta$ (or equivalently, the minimum distance) between them.'' \cite{delsarte_spherical_1991}. A solution for the spherical code problem can be written as $(m, N_s, \theta(m, N_s))$.
	The formal expression of the relation to the spherical code problem is shown as follows.
	\begin{theorem} \label{theorem_cluster_angle}
		Consider the convergence state for the co-evolutionary opinion dynamics with recommendation system in \eqref{OD_S} using the angle-based method with \eqref{s_g}, and \eqref{g_bar} with predefined $\epsilon \in [0,1]$.
		In particular, the system includes $n$ users with $m$-dimensional opinions.
		If there exist solutions for spherical code problem, $(m, N_s, \theta)$, where $\theta > \arccos(\epsilon)$. The number of the clustered opinion groups is upper bounded by $\min\{ N_s, n \}$.
	\end{theorem}
	\begin{proof}
		According to the definition of the sphere packing problem, $\theta$ is the maximum value of the minimum angle when placing $N_s$ points on the sphere. The following proof is by contradiction: if $\arccos(\epsilon) > \theta$ and it is possible to place $N_s$ points on the sphere with distances greater than $\arccos(\epsilon)$, then the maximum value of the minimum angle when placing $N_s$ points on the sphere should be $\arccos(\epsilon)$, which contradicts the original assumption. 
		
		Therefore, given the least pair-wise angle $\arccos(\epsilon) > \theta$, one can only place less than $N_s$ points on the sphere, which means, 
		the number of the clustered opinion groups is upper bounded by $N_s$. Considering also the number of users (points) in the graph as another upper bound, the proof is done.
	\end{proof}
	
	\cref{theorem_cluster_angle} shows the relationship between the clustering problem and the spherical codes problem. 
	Note that the solution to the sphere packing problem is difficult to express in a closed analytical form. However, numerous studies have proposed approximations. For instance, in~\cite{fejes1953lagerungen}, the \textit{Tóth bound} for the minimum distance \( d \) between any two points in an optimal configuration of \( N \) points on the unit sphere in three-dimensional space was derived as:
	\begin{equation}
		d_{\min} \leq \sqrt{4 - \csc^2 \left( \frac{N \pi}{6(N-2)} \right)}.
	\end{equation}
	This bound provides an analytical approximation of the minimum distance between points.
	
	Furthermore, more extensive results have been presented in~\cite{hamkins1996design, _spherical_}. Reference~\cite{_spherical_}, in particular, provides a comprehensive table summarizing all currently known and mathematically proven results. This table serves as the basis for the experimental setup discussed later in this paper.
	
	Until here, the performance of the autonomous ODRS system is studied.
	Next, the effect of the propagator is studied in the following section.
	
	\section{Opinion Manipulation in ODRS Framework}
	\label{section_control}
	
	In this section, the propagators are introduced to shape the user opinions within the proposed ODRS framework, which is formulated in \cref{subsection_propagator_formulation}.
	Then, the reinforcement learning algorithm is introduced in \cref{subsection_RL} to obtain the optimal influence strategy.
	
	\subsection{Formulation of Opinion Manipulation Problem}
	\label{subsection_propagator_formulation}
	
	In this subsection, $n_e \in \mathbb{N}$ external propagators are considered to guide the opinion evolution.
	The opinions for the propagators are denoted as $\bm{u}_i \in [0,1]^m$ for $i \in \mathcal{V}_e = \{ 1, \cdots, n_e \} \subset \mathbb{N}$.
	Considering the recommendation system treats the propagator opinions $\bm{u}_i$ identical as the user opinions $\bm{x}_j$ with $i \in \mathcal{V}_e$ and $j \in \mathcal{V}$, the affected opinion dynamics for the users is written as
	\begin{align}
		\bm{x}_i(k+1) = \sum_{j=1}^n w_{ij}^x(k) \bm{x}_j(k) + \sum_{j=1}^{n_e} w_{ij}^u(k) \bm{u}_j(k),
	\end{align}
	where the weights $w_{i,j}^x(k)$ and $w_{i,j}^u(k)$ are calculated as
	\begin{align}
		&w_{i,j}^x = \frac{s(\bm{x}_i, \bm{x}_j)}{\sum_{p = 1}^n s(\bm{x}_i, \bm{x}_p) + \sum_{p = 1}^{n_e} s(\bm{x}_i, \bm{u}_p)} , \\
		&w_{i,j}^u = \frac{s(\bm{x}_i, \bm{u}_j)}{\sum_{p = 1}^n s(\bm{x}_i, \bm{x}_p) + \sum_{p = 1}^{n_e} s(\bm{x}_i, \bm{u}_p)} 
	\end{align}
	following the normalization concept similar to that in \cref{eqn_w}.
	Concatenate the opinions of users and propagators as $\bm{X}$ and $\bm{U} = [\bm{u}_1^T, \cdots, \bm{u}_{n_e}^T] \in [0,1]^{n_e \times m}$ respectively, such that the concatenated opinion dynamics is written as
	\begin{align} \label{discrete model}
		\bm{X}(k+1) &= \bm{W}^x(\bm{X}(k)) \bm{X}(k) + \bm{W}^u(\bm{X}(k), \bm{U}(k)) \bm{U}(k)
	\end{align}
	with $\bm{W}^u(\bm{X}(k), \bm{U}(k)) = [w_{i,j}^u(k)]_{i \in \mathcal{V}, j \in \mathcal{V}_e} \in \mathbb{R}^{n \times n_e}$ and $\bm{W}^x(\bm{X}(k)) = [w_{i,j}^x(k)]_{i,j \in \mathcal{V}} \in \mathbb{R}^{n \times n}$.
	
	The propagators intend to guide the user opinions to a prescribed value denoted as $\bm{x}_c \in [0,1]^m$ with less effort, which is formulated as a receding horizon optimization problem with predictive horizon $N \in \mathbb{N}_+$ as
	\begin{align} \label{eqn_RL_optimization}
		&\min_{\bm{U}^N} J_N(\bm{U}^N) = J_N^x(\bm{X}^N) + J_N^u(\bm{U}^N)  \\
		\text{s.t.}~ &  \bm{X}(k+1) = \bm{W}^x(\bm{X}(k)) \bm{X}(k) + \bm{W}^u(\bm{X}(k), \bm{U}(k)) \bm{U}(k), \nonumber \\
		&\forall k = 0, \cdots,N-1 \nonumber
	\end{align}
	where $\bm{U}^N = \{ \bm{U}(0), \cdots, \bm{U}(N-1)\}$ and $\bm{X}^N = \{ \bm{X}(1), \cdots, \bm{X}(N)\}$.
	The opinion guidance cost $J_N^x(\bm{X}^N)$ and influences cost $J_N^u(\bm{U}^N)$ are written as
	\begin{align}
		&J_N^x(\bm{X}^N) = \sum\nolimits_{i=1}^n \sum\nolimits_{k=1}^N  (\bm{x}_i(k) -\bm{x}^c)^T \bm{P}_x (\bm{x}_i(k) - \bm{x}^c), \nonumber \\
		&J_N^u(\bm{U}^N) = \sum\nolimits_{i=1}^{n_e} \sum\nolimits_{k=0}^{N-1}  \bm{u}_i^T(k) \bm{P}_u \bm{u}_i(k) 
	\end{align}
	with given coefficient matrices $\bm{P}_x, \bm{P}_u \in \mathbb{R}^{m \times m}$.
	While the cost function $J_N(\bm{U}^N)$ is in the quadratic form, the non-linearity and discontinuity in $\bm{W}^x(\bm{X}(k))$ and $\bm{W}^u(\bm{X}(k), \bm{U}(k))$ induces large challenge to obtain the analytical solution of \eqref{eqn_RL_optimization}.
    There exist some heuristic methods to solve such a complex optimal control problem \cite{gao2021quasi}, but their solution is fixed for a specific setting.
	Moreover, an adaptive influence strategy is required for different initial opinions $\bm{X}(0)$ and desired opinions $\bm{x}_c$, such that the reinforcement learning becomes a suitable solving tool for \eqref{eqn_RL_optimization}.
	
	\subsection{Reinforcement Learning for Optimal Opinion Manipulation}
	\label{subsection_RL}
	
	Considering the strong non-linearity induced by the recommendation system, solving the optimization problem in \eqref{eqn_RL_optimization} with conventional methods, e.g., gradient descent strategy, may result in a sub-optimal solution with merely a local minimum.
	To address this issue, deep reinforcement learning is employed, taking both exploration and exploitation into consideration and aiming to find the global optimal solution.
	Specifically, the overall structure of the reinforcement learning framework is shown in \cref{fig_RL_structure}.
	The environment composed of opinion dynamics and recommendation system takes $\bm{U}(k)$ as input, and returns the user opinion in the next step, i.e., $\bm{X}(k+1)$, and the reward function $r(k)$ is defined as
	\begin{align}
		&r(k) = - \mathrm{tr}(\bm{E}(k) \bm{P}_x \bm{E}^T(k)) + \mathrm{tr}(\bm{U}(k - 1) \bm{P}_u \bm{U}^T(k - 1) ), \nonumber \\
		&\bm{E}(k) = \bm{X}(k) - (\bm{1}_n \otimes \bm{x}^c)^T
	\end{align}
	for given $\bm{x}^c$.
	It is obvious to see $J_N(\bm{U}^N)$ in \eqref{eqn_RL_optimization} is the negative cumulative reward of $r(k)$, i.e., $J_N(\bm{U}^N) = - \sum_{k = 1}^N r(k)$.
	The agent for control input generation uses the proximal policy optimization (PPO) structure, which follows the actor-critic framework.
	After the reshape layer converting the matrix input $\bm{X}(k)$ to a vector, the optimal accumulated reward $V(k) = \sum_{k=1}^N r(k)$ from state $\bm{X}(k)$ is estimated, which serves guidance for training the actor-network.
	The actor-network generates the control input $\bm{U}(k)$ based on the information of current state $\bm{X}(k)$ and estimated accumulated reward $V(k)$.
	The detailed design of neural networks in actor and critic networks is shown in \cref{section_experiment}.
	
	\begin{figure}[t]
		\centering
		\includegraphics[width = 0.48\textwidth]{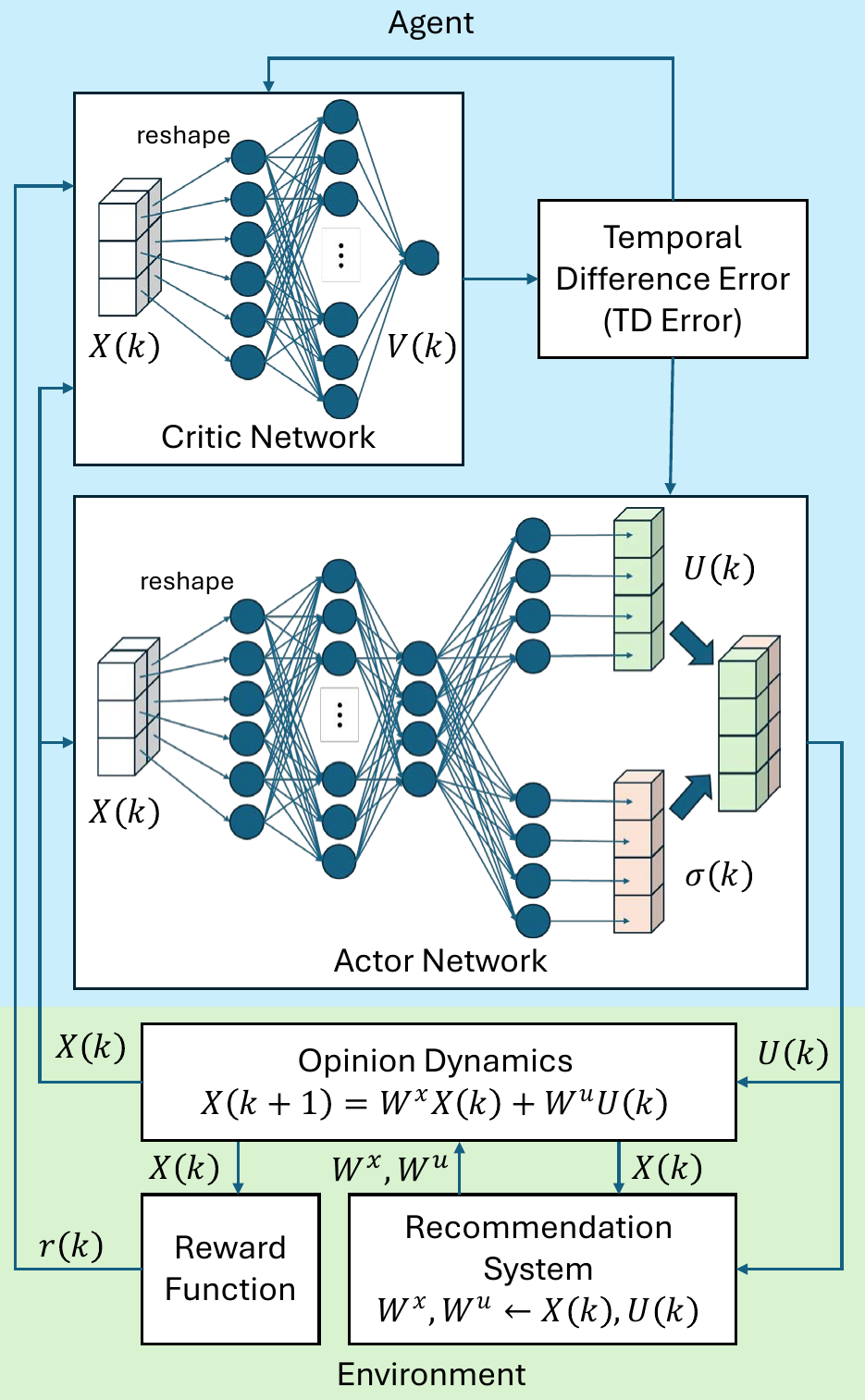}
		\caption{Structure of reinforcement learning with PPO agent for opinion dynamics with recommendation system}
		\label{fig_RL_structure}
	\end{figure}

    \begin{remark}
        In this paper, the actions of all propagators are coordinated through a centralized RL framework to pursue a common target. 
        In practice, however, different propagators may aim at distinct objectives \cite{chen2025modeling}, e.g., during elections where they support different candidates. 
        In such cases, each propagator must make decisions based only on user information. 
        The optimal individual actions can be studied through a game-theoretic RL perspective, where solutions correspond to a Nash equilibrium. 
        Furthermore, since in reality propagators often operate asynchronously, the problem can also be formulated as a distributed optimal control task \cite{chen2025distributed} and addressed using sequential RL.
    \end{remark}
    
	\section{Simulation}
	
	This chapter uses simulations to validate the preceding analysis and to demonstrate the effectiveness of the reinforcement-learning-based control strategy.
	
	\label{section_experiment}
	
	\subsection{Validation of the Analysis}
	
	\subsubsection{Experiment Setup}
	
	The experiments are conducted using MATLAB simulations based on the Yelp dataset\cite{_yelp_}. Yelp is a platform for discovering and reviewing local businesses. We use a subset of the Yelp dataset, focusing on user IDs, business IDs, and ratings, which serve as the users, items, and initial values in our opinion matrix.
	Based on this real data, we carry out simulations under different parameter settings in both distance-based and angle-based frameworks.

	\subsubsection{Distance-Based Analysis}
	
	In the system evolution experiments, we use 3 items and 14 users. From the results in \cref{fig:distance}, we observe that regardless of the bound settings, the system always reaches convergence. When \( \sqrt{m}(1-\epsilon) = 0.1 \), only points within a distance of $0.1$ are connected. Due to the discrete nature of the Yelp ratings (0, 0.25, 0.75, 1), if \( \sqrt{m}(1-\epsilon) \) is set below $0.25$, the network remains unchanged since the minimum distance between any two distinct vectors is larger than 0.25. Consequently, users do not form connections, and their opinions remain unchanged.

	\begin{figure*}[t]
		\centering
		\begin{subfigure}{1\textwidth}
			\begin{tikzpicture}
				\def\file{pictures/dis_ep0.1_op1.txt}
				\begin{axis}[xlabel={$k$},ylabel={$x_i^1$},
					xmin=0.2, ymin = -0.1, xmax = 9.8,ymax=1.1,legend columns=14,
					width=0.35\textwidth,height=3.5cm,legend style={at={(3.42,1.15)},anchor=east}]
					
					\addplot[red!80, thick, draw opacity = 0.5]      table[x = _0 , y  = _1 ]{\file};
					\addplot[blue!80, thick, draw opacity = 0.5]      table[x = _0 , y  = _2 ]{\file};
					\addplot[green!80, thick, draw opacity = 0.5]      table[x = _0 , y  = _3 ]{\file};
					\addplot[cyan!80, thick, draw opacity = 0.5]      table[x = _0, y  = _4 ]{\file};
					\addplot[pink!80, thick, draw opacity = 0.5]      table[x = _0 , y  = _5 ]{\file};
					\addplot[teal!80, thick, draw opacity = 0.5]      table[x = _0 , y  = _6 ]{\file};
					\addplot[violet!80, thick, draw opacity = 0.5]      table[x = _0 , y  = _7 ]{\file};
					\addplot[red!80, dashed, thick, draw opacity = 0.5]      table[x = _0 , y  = _8 ]{\file};
					\addplot[blue!80, dashed, thick, draw opacity = 0.5]      table[x = _0 , y  = _9 ]{\file};
					\addplot[green!80, dashed, thick, draw opacity = 0.5]      table[x = _0 , y  = _10 ]{\file};
					\addplot[cyan!80, dashed, thick, draw opacity = 0.5]      table[x = _0 , y  = _11 ]{\file};
					\addplot[pink!80, dashed, thick, draw opacity = 0.5]      table[x = _0 , y  = _12 ]{\file};
					\addplot[teal!80, dashed, thick, draw opacity = 0.5]      table[x = _0 , y  = _13 ]{\file};
					\addplot[violet!80, dashed, thick, draw opacity = 0.5]      table[x = _0 , y  = _14 ]{\file};
					
					\legend{
						$\!\! \bm{x}_1 \!\!$, $\!\! \bm{x}_2 \!\!$, $\!\! \bm{x}_3 \!\!$, $\!\! \bm{x}_4 \!\!$, $\!\! \bm{x}_5 \!\!$, $\!\! \bm{x}_6 \!\!$, $\!\! \bm{x}_7 \!\!$, 
						$\!\! \bm{x}_8 \!\!$, $\!\! \bm{x}_9 \!\!$, $\!\! \bm{x}_{10} \!\!$, $\!\! \bm{x}_{11} \!\!$, $\!\! \bm{x}_{12} \!\!$, $\!\! \bm{x}_{13} \!\!$, $\!\! \bm{x}_{14} \!\!$}
				\end{axis}
				\def\file{pictures/dis_ep0.1_op2.txt}
				\begin{axis}[xlabel={$k$},ylabel={$x_i^2$},
					xmin=0.2, ymin = -0.1, xmax = 9.8,ymax=1.1,legend columns=7,
					width=0.35\textwidth,height=3.5cm,legend style={at={(1,1.25)},anchor=east},
					ylabel shift=-0.2cm,
					xshift=5.8cm]
					
					\addplot[red!80, thick, draw opacity = 0.5]      table[x = _0 , y  = _1 ]{\file};
					\addplot[blue!80, thick, draw opacity = 0.5]      table[x = _0 , y  = _2 ]{\file};
					\addplot[green!80, thick, draw opacity = 0.5]      table[x = _0 , y  = _3 ]{\file};
					\addplot[cyan!80, thick, draw opacity = 0.5]      table[x = _0, y  = _4 ]{\file};
					\addplot[pink!80, thick, draw opacity = 0.5]      table[x = _0 , y  = _5 ]{\file};
					\addplot[teal!80, thick, draw opacity = 0.5]      table[x = _0 , y  = _6 ]{\file};
					\addplot[violet!80, thick, draw opacity = 0.5]      table[x = _0 , y  = _7 ]{\file};
					\addplot[red!80, dashed, thick, draw opacity = 0.5]      table[x = _0 , y  = _8 ]{\file};
					\addplot[blue!80, dashed, thick, draw opacity = 0.5]      table[x = _0 , y  = _9 ]{\file};
					\addplot[green!80, dashed, thick, draw opacity = 0.5]      table[x = _0 , y  = _10 ]{\file};
					\addplot[cyan!80, dashed, thick, draw opacity = 0.5]      table[x = _0 , y  = _11 ]{\file};
					\addplot[pink!80, dashed, thick, draw opacity = 0.5]      table[x = _0 , y  = _12 ]{\file};
					\addplot[teal!80, dashed, thick, draw opacity = 0.5]      table[x = _0 , y  = _13 ]{\file};
					\addplot[violet!80, dashed, thick, draw opacity = 0.5]      table[x = _0 , y  = _14 ]{\file};
					
				\end{axis}
				\def\file{pictures/dis_ep0.1_op3.txt}
				\begin{axis}[xlabel={$k$},ylabel={$x_i^2$},
					xmin=0.2, ymin = -0.1, xmax = 9.8,ymax=1.1,legend columns=7,
					width=0.35\textwidth,height=3.5cm,legend style={at={(1,1.25)},anchor=east},
					ylabel shift=-0.2cm,
					xshift=11.6cm]
					
					\addplot[red!80, thick, draw opacity = 0.5]      table[x = _0 , y  = _1 ]{\file};
					\addplot[blue!80, thick, draw opacity = 0.5]      table[x = _0 , y  = _2 ]{\file};
					\addplot[green!80, thick, draw opacity = 0.5]      table[x = _0 , y  = _3 ]{\file};
					\addplot[cyan!80, thick, draw opacity = 0.5]      table[x = _0, y  = _4 ]{\file};
					\addplot[pink!80, thick, draw opacity = 0.5]      table[x = _0 , y  = _5 ]{\file};
					\addplot[teal!80, thick, draw opacity = 0.5]      table[x = _0 , y  = _6 ]{\file};
					\addplot[violet!80, thick, draw opacity = 0.5]      table[x = _0 , y  = _7 ]{\file};
					\addplot[red!80, dashed, thick, draw opacity = 0.5]      table[x = _0 , y  = _8 ]{\file};
					\addplot[blue!80, dashed, thick, draw opacity = 0.5]      table[x = _0 , y  = _9 ]{\file};
					\addplot[green!80, dashed, thick, draw opacity = 0.5]      table[x = _0 , y  = _10 ]{\file};
					\addplot[cyan!80, dashed, thick, draw opacity = 0.5]      table[x = _0 , y  = _11 ]{\file};
					\addplot[pink!80, dashed, thick, draw opacity = 0.5]      table[x = _0 , y  = _12 ]{\file};
					\addplot[teal!80, dashed, thick, draw opacity = 0.5]      table[x = _0 , y  = _13 ]{\file};
					\addplot[violet!80, dashed, thick, draw opacity = 0.5]      table[x = _0 , y  = _14 ]{\file};
					
				\end{axis}
			\end{tikzpicture}
			\vspace{-0.3cm}
			\caption{
				$\sqrt{m}(1-\epsilon) = 0.1$
			}
		\end{subfigure}
		\begin{subfigure}{1\textwidth}
			\begin{tikzpicture}
				\def\file{pictures/dis_ep0.4_op1.txt}
				\begin{axis}[xlabel={$k$},ylabel={$x_i^1$},
					xmin=0.2, ymin = -0.1, xmax = 9.8,ymax=1.1,legend columns=14,
					width=0.35\textwidth,height=3.5cm,legend style={at={(3.42,1.15)},anchor=east}]
					
					\addplot[red!80, thick, draw opacity = 0.5]      table[x = _0 , y  = _1 ]{\file};
					\addplot[blue!80, thick, draw opacity = 0.5]      table[x = _0 , y  = _2 ]{\file};
					\addplot[green!80, thick, draw opacity = 0.5]      table[x = _0 , y  = _3 ]{\file};
					\addplot[cyan!80, thick, draw opacity = 0.5]      table[x = _0, y  = _4 ]{\file};
					\addplot[pink!80, thick, draw opacity = 0.5]      table[x = _0 , y  = _5 ]{\file};
					\addplot[teal!80, thick, draw opacity = 0.5]      table[x = _0 , y  = _6 ]{\file};
					\addplot[violet!80, thick, draw opacity = 0.5]      table[x = _0 , y  = _7 ]{\file};
					\addplot[red!80, dashed, thick, draw opacity = 0.5]      table[x = _0 , y  = _8 ]{\file};
					\addplot[blue!80, dashed, thick, draw opacity = 0.5]      table[x = _0 , y  = _9 ]{\file};
					\addplot[green!80, dashed, thick, draw opacity = 0.5]      table[x = _0 , y  = _10 ]{\file};
					\addplot[cyan!80, dashed, thick, draw opacity = 0.5]      table[x = _0 , y  = _11 ]{\file};
					\addplot[pink!80, dashed, thick, draw opacity = 0.5]      table[x = _0 , y  = _12 ]{\file};
					\addplot[teal!80, dashed, thick, draw opacity = 0.5]      table[x = _0 , y  = _13 ]{\file};
					\addplot[violet!80, dashed, thick, draw opacity = 0.5]      table[x = _0 , y  = _14 ]{\file};
					
					\legend{
						$\!\! \bm{x}_1 \!\!$, $\!\! \bm{x}_2 \!\!$, $\!\! \bm{x}_3 \!\!$, $\!\! \bm{x}_4 \!\!$, $\!\! \bm{x}_5 \!\!$, $\!\! \bm{x}_6 \!\!$, $\!\! \bm{x}_7 \!\!$, 
						$\!\! \bm{x}_8 \!\!$, $\!\! \bm{x}_9 \!\!$, $\!\! \bm{x}_{10} \!\!$, $\!\! \bm{x}_{11} \!\!$, $\!\! \bm{x}_{12} \!\!$, $\!\! \bm{x}_{13} \!\!$, $\!\! \bm{x}_{14} \!\!$}
				\end{axis}
				\def\file{pictures/dis_ep0.4_op2.txt}
				\begin{axis}[xlabel={$k$},ylabel={$x_i^2$},
					xmin=0.2, ymin = -0.1, xmax = 9.8,ymax=1.1,legend columns=7,
					width=0.35\textwidth,height=3.5cm,legend style={at={(1,1.25)},anchor=east},
					ylabel shift=-0.2cm,
					xshift=5.8cm]
					
					\addplot[red!80, thick, draw opacity = 0.5]      table[x = _0 , y  = _1 ]{\file};
					\addplot[blue!80, thick, draw opacity = 0.5]      table[x = _0 , y  = _2 ]{\file};
					\addplot[green!80, thick, draw opacity = 0.5]      table[x = _0 , y  = _3 ]{\file};
					\addplot[cyan!80, thick, draw opacity = 0.5]      table[x = _0, y  = _4 ]{\file};
					\addplot[pink!80, thick, draw opacity = 0.5]      table[x = _0 , y  = _5 ]{\file};
					\addplot[teal!80, thick, draw opacity = 0.5]      table[x = _0 , y  = _6 ]{\file};
					\addplot[violet!80, thick, draw opacity = 0.5]      table[x = _0 , y  = _7 ]{\file};
					\addplot[red!80, dashed, thick, draw opacity = 0.5]      table[x = _0 , y  = _8 ]{\file};
					\addplot[blue!80, dashed, thick, draw opacity = 0.5]      table[x = _0 , y  = _9 ]{\file};
					\addplot[green!80, dashed, thick, draw opacity = 0.5]      table[x = _0 , y  = _10 ]{\file};
					\addplot[cyan!80, dashed, thick, draw opacity = 0.5]      table[x = _0 , y  = _11 ]{\file};
					\addplot[pink!80, dashed, thick, draw opacity = 0.5]      table[x = _0 , y  = _12 ]{\file};
					\addplot[teal!80, dashed, thick, draw opacity = 0.5]      table[x = _0 , y  = _13 ]{\file};
					\addplot[violet!80, dashed, thick, draw opacity = 0.5]      table[x = _0 , y  = _14 ]{\file};
					
				\end{axis}
				\def\file{pictures/dis_ep0.4_op3.txt}
				\begin{axis}[xlabel={$k$},ylabel={$x_i^2$},
					xmin=0.2, ymin = -0.1, xmax = 9.8,ymax=1.1,legend columns=7,
					width=0.35\textwidth,height=3.5cm,legend style={at={(1,1.25)},anchor=east},
					ylabel shift=-0.2cm,
					xshift=11.6cm]
					
					\addplot[red!80, thick, draw opacity = 0.5]      table[x = _0 , y  = _1 ]{\file};
					\addplot[blue!80, thick, draw opacity = 0.5]      table[x = _0 , y  = _2 ]{\file};
					\addplot[green!80, thick, draw opacity = 0.5]      table[x = _0 , y  = _3 ]{\file};
					\addplot[cyan!80, thick, draw opacity = 0.5]      table[x = _0, y  = _4 ]{\file};
					\addplot[pink!80, thick, draw opacity = 0.5]      table[x = _0 , y  = _5 ]{\file};
					\addplot[teal!80, thick, draw opacity = 0.5]      table[x = _0 , y  = _6 ]{\file};
					\addplot[violet!80, thick, draw opacity = 0.5]      table[x = _0 , y  = _7 ]{\file};
					\addplot[red!80, dashed, thick, draw opacity = 0.5]      table[x = _0 , y  = _8 ]{\file};
					\addplot[blue!80, dashed, thick, draw opacity = 0.5]      table[x = _0 , y  = _9 ]{\file};
					\addplot[green!80, dashed, thick, draw opacity = 0.5]      table[x = _0 , y  = _10 ]{\file};
					\addplot[cyan!80, dashed, thick, draw opacity = 0.5]      table[x = _0 , y  = _11 ]{\file};
					\addplot[pink!80, dashed, thick, draw opacity = 0.5]      table[x = _0 , y  = _12 ]{\file};
					\addplot[teal!80, dashed, thick, draw opacity = 0.5]      table[x = _0 , y  = _13 ]{\file};
					\addplot[violet!80, dashed, thick, draw opacity = 0.5]      table[x = _0 , y  = _14 ]{\file};
					
				\end{axis}
			\end{tikzpicture}
			\vspace{-0.3cm}
			\caption{
				$\sqrt{m}(1-\epsilon) = 0.4$
			}
		\end{subfigure}
		\begin{subfigure}{1\textwidth}
			\begin{tikzpicture}
				\def\file{pictures/dis_ep0.9_op1.txt}
				\begin{axis}[xlabel={$k$},ylabel={$x_i^1$},
					xmin=0.2, ymin = -0.1, xmax = 9.8,ymax=1.1,legend columns=14,
					width=0.35\textwidth,height=3.5cm,legend style={at={(3.42,1.15)},anchor=east}]
					
					\addplot[red!80, thick, draw opacity = 0.5]      table[x = _0 , y  = _1 ]{\file};
					\addplot[blue!80, thick, draw opacity = 0.5]      table[x = _0 , y  = _2 ]{\file};
					\addplot[green!80, thick, draw opacity = 0.5]      table[x = _0 , y  = _3 ]{\file};
					\addplot[cyan!80, thick, draw opacity = 0.5]      table[x = _0, y  = _4 ]{\file};
					\addplot[pink!80, thick, draw opacity = 0.5]      table[x = _0 , y  = _5 ]{\file};
					\addplot[teal!80, thick, draw opacity = 0.5]      table[x = _0 , y  = _6 ]{\file};
					\addplot[violet!80, thick, draw opacity = 0.5]      table[x = _0 , y  = _7 ]{\file};
					\addplot[red!80, dashed, thick, draw opacity = 0.5]      table[x = _0 , y  = _8 ]{\file};
					\addplot[blue!80, dashed, thick, draw opacity = 0.5]      table[x = _0 , y  = _9 ]{\file};
					\addplot[green!80, dashed, thick, draw opacity = 0.5]      table[x = _0 , y  = _10 ]{\file};
					\addplot[cyan!80, dashed, thick, draw opacity = 0.5]      table[x = _0 , y  = _11 ]{\file};
					\addplot[pink!80, dashed, thick, draw opacity = 0.5]      table[x = _0 , y  = _12 ]{\file};
					\addplot[teal!80, dashed, thick, draw opacity = 0.5]      table[x = _0 , y  = _13 ]{\file};
					\addplot[violet!80, dashed, thick, draw opacity = 0.5]      table[x = _0 , y  = _14 ]{\file};
					
					\legend{
						$\!\! \bm{x}_1 \!\!$, $\!\! \bm{x}_2 \!\!$, $\!\! \bm{x}_3 \!\!$, $\!\! \bm{x}_4 \!\!$, $\!\! \bm{x}_5 \!\!$, $\!\! \bm{x}_6 \!\!$, $\!\! \bm{x}_7 \!\!$, 
						$\!\! \bm{x}_8 \!\!$, $\!\! \bm{x}_9 \!\!$, $\!\! \bm{x}_{10} \!\!$, $\!\! \bm{x}_{11} \!\!$, $\!\! \bm{x}_{12} \!\!$, $\!\! \bm{x}_{13} \!\!$, $\!\! \bm{x}_{14} \!\!$}
				\end{axis}
				\def\file{pictures/dis_ep0.9_op1.txt}
				\begin{axis}[xlabel={$k$},ylabel={$x_i^2$},
					xmin=0.2, ymin = -0.1, xmax = 9.8,ymax=1.1,legend columns=7,
					width=0.35\textwidth,height=3.5cm,legend style={at={(1,1.25)},anchor=east},
					ylabel shift=-0.2cm,
					xshift=5.8cm]
					
					\addplot[red!80, thick, draw opacity = 0.5]      table[x = _0 , y  = _1 ]{\file};
					\addplot[blue!80, thick, draw opacity = 0.5]      table[x = _0 , y  = _2 ]{\file};
					\addplot[green!80, thick, draw opacity = 0.5]      table[x = _0 , y  = _3 ]{\file};
					\addplot[cyan!80, thick, draw opacity = 0.5]      table[x = _0, y  = _4 ]{\file};
					\addplot[pink!80, thick, draw opacity = 0.5]      table[x = _0 , y  = _5 ]{\file};
					\addplot[teal!80, thick, draw opacity = 0.5]      table[x = _0 , y  = _6 ]{\file};
					\addplot[violet!80, thick, draw opacity = 0.5]      table[x = _0 , y  = _7 ]{\file};
					\addplot[red!80, dashed, thick, draw opacity = 0.5]      table[x = _0 , y  = _8 ]{\file};
					\addplot[blue!80, dashed, thick, draw opacity = 0.5]      table[x = _0 , y  = _9 ]{\file};
					\addplot[green!80, dashed, thick, draw opacity = 0.5]      table[x = _0 , y  = _10 ]{\file};
					\addplot[cyan!80, dashed, thick, draw opacity = 0.5]      table[x = _0 , y  = _11 ]{\file};
					\addplot[pink!80, dashed, thick, draw opacity = 0.5]      table[x = _0 , y  = _12 ]{\file};
					\addplot[teal!80, dashed, thick, draw opacity = 0.5]      table[x = _0 , y  = _13 ]{\file};
					\addplot[violet!80, dashed, thick, draw opacity = 0.5]      table[x = _0 , y  = _14 ]{\file};
					
				\end{axis}
				\def\file{pictures/dis_ep0.9_op1.txt}
				\begin{axis}[xlabel={$k$},ylabel={$x_i^2$},
					xmin=0.2, ymin = -0.1, xmax = 9.8,ymax=1.1,legend columns=7,
					width=0.35\textwidth,height=3.5cm,legend style={at={(1,1.25)},anchor=east},
					ylabel shift=-0.2cm,
					xshift=11.6cm]
					
					\addplot[red!80, thick, draw opacity = 0.5]      table[x = _0 , y  = _1 ]{\file};
					\addplot[blue!80, thick, draw opacity = 0.5]      table[x = _0 , y  = _2 ]{\file};
					\addplot[green!80, thick, draw opacity = 0.5]      table[x = _0 , y  = _3 ]{\file};
					\addplot[cyan!80, thick, draw opacity = 0.5]      table[x = _0, y  = _4 ]{\file};
					\addplot[pink!80, thick, draw opacity = 0.5]      table[x = _0 , y  = _5 ]{\file};
					\addplot[teal!80, thick, draw opacity = 0.5]      table[x = _0 , y  = _6 ]{\file};
					\addplot[violet!80, thick, draw opacity = 0.5]      table[x = _0 , y  = _7 ]{\file};
					\addplot[red!80, dashed, thick, draw opacity = 0.5]      table[x = _0 , y  = _8 ]{\file};
					\addplot[blue!80, dashed, thick, draw opacity = 0.5]      table[x = _0 , y  = _9 ]{\file};
					\addplot[green!80, dashed, thick, draw opacity = 0.5]      table[x = _0 , y  = _10 ]{\file};
					\addplot[cyan!80, dashed, thick, draw opacity = 0.5]      table[x = _0 , y  = _11 ]{\file};
					\addplot[pink!80, dashed, thick, draw opacity = 0.5]      table[x = _0 , y  = _12 ]{\file};
					\addplot[teal!80, dashed, thick, draw opacity = 0.5]      table[x = _0 , y  = _13 ]{\file};
					\addplot[violet!80, dashed, thick, draw opacity = 0.5]      table[x = _0 , y  = _14 ]{\file};
					
				\end{axis}
			\end{tikzpicture}
			\vspace{-0.3cm}
			\caption{
				$\sqrt{m}(1-\epsilon) = 0.9$
			}
		\end{subfigure}
		\vspace{-0.3cm}
		\caption{
			Distance-based method
		}
		\label{fig:distance}
	\end{figure*}
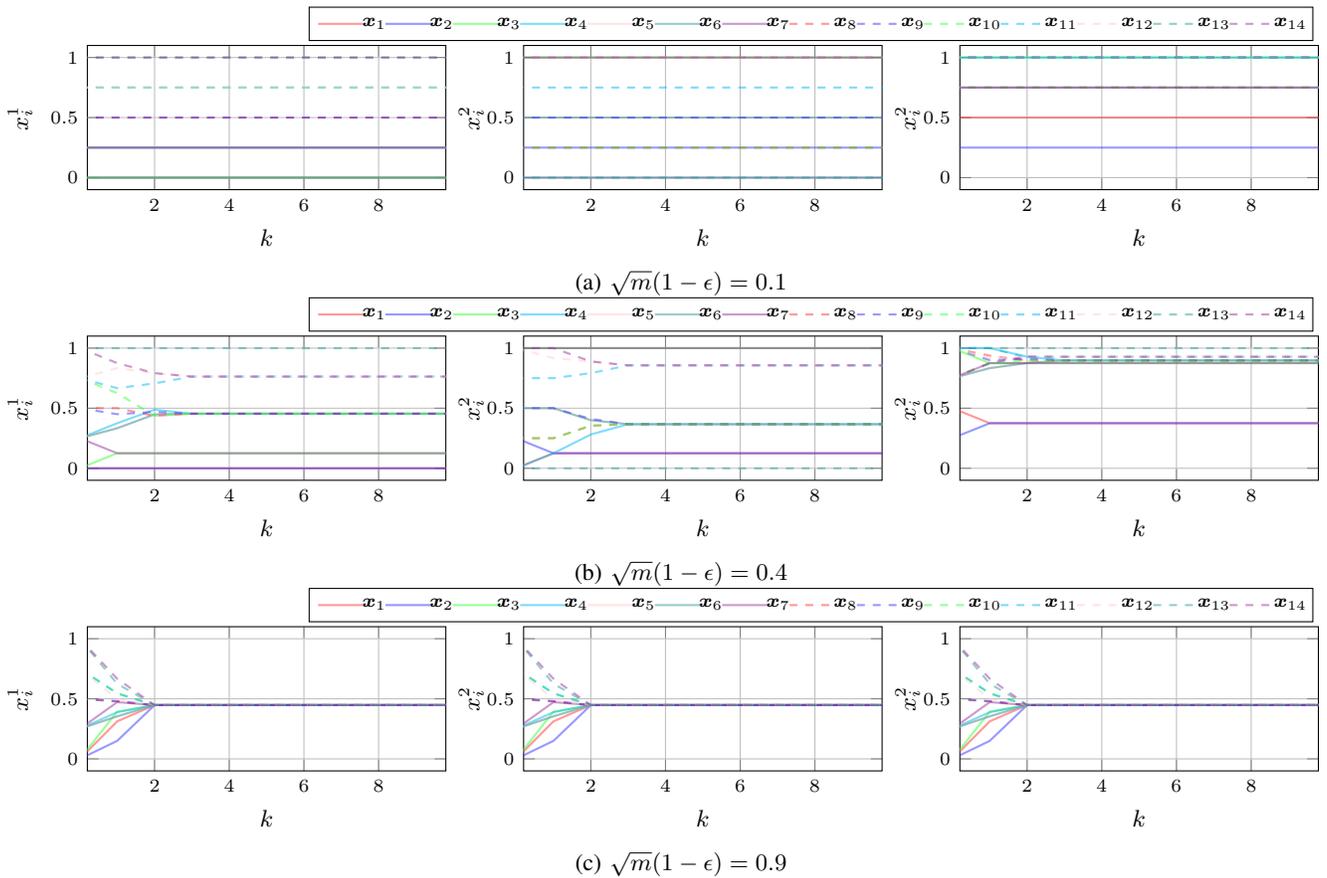
	
	When $\sqrt{m}(1-\epsilon) = 0.4$, points within a distance of $0.4$ are connected, while those beyond $0.4$ are not connected. The evolving connection graph results in multiple connected components, each achieving internal consensus.
	
	When $\sqrt{m}(1-\epsilon) = 0.9$, only points beyond a distance of $0.9$ are disconnected, resulting in an almost fully connected graph, and the entire system reaches consensus.
	
	Additionally, we plot the relationship between $\epsilon$ and the number of clusters in \cref{figure_N_epsilon_dis}, increasing the number of users to $50$ to explore clustering more thoroughly. It is evident that the number of clusters increases as the threshold for connection increases. \cref{figure_N_epsilon_dis} also indicates that our proposed expression effectively describes the upper bound of the clustering result.
	
	\begin{figure}[t] 
		\centering
		\def\file{pictures/dist_group1.txt}
		\begin{tikzpicture}
			\begin{axis}[xlabel={$\epsilon$},ylabel={Nr. Groups},
				xmin=0.01, ymin = -0.1, xmax = 0.99,ymax=49,legend columns=1,
				width=0.48\textwidth,height=3.5cm,legend pos= north west
				]
				
				\addplot[blue!80, only marks, mark=o, mark size=1]      table[x = _1 , y  = _2 ]{\file};
				\addplot[red!80, only marks, mark=o, mark size=1]      table[x = _1 , y  = _3 ]{\file};
				
				\legend{Experimental result, Theoretical upper bound}
			\end{axis}
		\end{tikzpicture}
		\vspace{-0.3cm}
		\caption{
			Number of clusters as $\epsilon$ varies in the proposed ODRS framework using the distance-based method.}
		\vspace{-0.3cm}
		\label{figure_N_epsilon_dis}
	\end{figure}
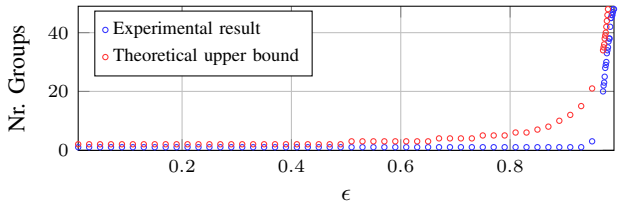
	
	\subsubsection{Angle-Based Analysis}
	
	The results for the angle-based model are generally similar to those of the distance-based model. To demonstrate the behavior of clustering and consensus, the results for $\epsilon = 0.93$, $0.97$, $0.99$ are shown in \cref{fig:angle}.
	
	\begin{figure*}[t]
		\centering
		\begin{subfigure}{1\textwidth}
			\begin{tikzpicture}
				\def\file{pictures/ang_ep0.93_op1.txt}
				\begin{axis}[xlabel={$k$},ylabel={$x_i^1$},
					xmin=0.2, ymin = -0.1, xmax = 9.8,ymax=1.1,legend columns=14,
					width=0.35\textwidth,height=3.5cm,legend style={at={(3.42,1.15)},anchor=east}]
					
					\addplot[red!80, thick, draw opacity = 0.5]      table[x = _0 , y  = _1 ]{\file};
					\addplot[blue!80, thick, draw opacity = 0.5]      table[x = _0 , y  = _2 ]{\file};
					\addplot[green!80, thick, draw opacity = 0.5]      table[x = _0 , y  = _3 ]{\file};
					\addplot[cyan!80, thick, draw opacity = 0.5]      table[x = _0, y  = _4 ]{\file};
					\addplot[pink!80, thick, draw opacity = 0.5]      table[x = _0 , y  = _5 ]{\file};
					\addplot[teal!80, thick, draw opacity = 0.5]      table[x = _0 , y  = _6 ]{\file};
					\addplot[violet!80, thick, draw opacity = 0.5]      table[x = _0 , y  = _7 ]{\file};
					\addplot[red!80, dashed, thick, draw opacity = 0.5]      table[x = _0 , y  = _8 ]{\file};
					\addplot[blue!80, dashed, thick, draw opacity = 0.5]      table[x = _0 , y  = _9 ]{\file};
					\addplot[green!80, dashed, thick, draw opacity = 0.5]      table[x = _0 , y  = _10 ]{\file};
					\addplot[cyan!80, dashed, thick, draw opacity = 0.5]      table[x = _0 , y  = _11 ]{\file};
					\addplot[pink!80, dashed, thick, draw opacity = 0.5]      table[x = _0 , y  = _12 ]{\file};
					\addplot[teal!80, dashed, thick, draw opacity = 0.5]      table[x = _0 , y  = _13 ]{\file};
					\addplot[violet!80, dashed, thick, draw opacity = 0.5]      table[x = _0 , y  = _14 ]{\file};
					
					\legend{
						$\!\! \bm{x}_1 \!\!$, $\!\! \bm{x}_2 \!\!$, $\!\! \bm{x}_3 \!\!$, $\!\! \bm{x}_4 \!\!$, $\!\! \bm{x}_5 \!\!$, $\!\! \bm{x}_6 \!\!$, $\!\! \bm{x}_7 \!\!$, 
						$\!\! \bm{x}_8 \!\!$, $\!\! \bm{x}_9 \!\!$, $\!\! \bm{x}_{10} \!\!$, $\!\! \bm{x}_{11} \!\!$, $\!\! \bm{x}_{12} \!\!$, $\!\! \bm{x}_{13} \!\!$, $\!\! \bm{x}_{14} \!\!$}
				\end{axis}
				\def\file{pictures/ang_ep0.93_op2.txt}
				\begin{axis}[xlabel={$k$},ylabel={$x_i^2$},
					xmin=0.2, ymin = -0.1, xmax = 9.8,ymax=1.1,legend columns=7,
					width=0.35\textwidth,height=3.5cm,legend style={at={(1,1.25)},anchor=east},
					ylabel shift=-0.2cm,
					xshift=5.8cm]
					
					\addplot[red!80, thick, draw opacity = 0.5]      table[x = _0 , y  = _1 ]{\file};
					\addplot[blue!80, thick, draw opacity = 0.5]      table[x = _0 , y  = _2 ]{\file};
					\addplot[green!80, thick, draw opacity = 0.5]      table[x = _0 , y  = _3 ]{\file};
					\addplot[cyan!80, thick, draw opacity = 0.5]      table[x = _0, y  = _4 ]{\file};
					\addplot[pink!80, thick, draw opacity = 0.5]      table[x = _0 , y  = _5 ]{\file};
					\addplot[teal!80, thick, draw opacity = 0.5]      table[x = _0 , y  = _6 ]{\file};
					\addplot[violet!80, thick, draw opacity = 0.5]      table[x = _0 , y  = _7 ]{\file};
					\addplot[red!80, dashed, thick, draw opacity = 0.5]      table[x = _0 , y  = _8 ]{\file};
					\addplot[blue!80, dashed, thick, draw opacity = 0.5]      table[x = _0 , y  = _9 ]{\file};
					\addplot[green!80, dashed, thick, draw opacity = 0.5]      table[x = _0 , y  = _10 ]{\file};
					\addplot[cyan!80, dashed, thick, draw opacity = 0.5]      table[x = _0 , y  = _11 ]{\file};
					\addplot[pink!80, dashed, thick, draw opacity = 0.5]      table[x = _0 , y  = _12 ]{\file};
					\addplot[teal!80, dashed, thick, draw opacity = 0.5]      table[x = _0 , y  = _13 ]{\file};
					\addplot[violet!80, dashed, thick, draw opacity = 0.5]      table[x = _0 , y  = _14 ]{\file};
					
				\end{axis}
				\def\file{pictures/ang_ep0.93_op3.txt}
				\begin{axis}[xlabel={$k$},ylabel={$x_i^2$},
					xmin=0.2, ymin = -0.1, xmax = 9.8,ymax=1.1,legend columns=7,
					width=0.35\textwidth,height=3.5cm,legend style={at={(1,1.25)},anchor=east},
					ylabel shift=-0.2cm,
					xshift=11.6cm]
					
					\addplot[red!80, thick, draw opacity = 0.5]      table[x = _0 , y  = _1 ]{\file};
					\addplot[blue!80, thick, draw opacity = 0.5]      table[x = _0 , y  = _2 ]{\file};
					\addplot[green!80, thick, draw opacity = 0.5]      table[x = _0 , y  = _3 ]{\file};
					\addplot[cyan!80, thick, draw opacity = 0.5]      table[x = _0, y  = _4 ]{\file};
					\addplot[pink!80, thick, draw opacity = 0.5]      table[x = _0 , y  = _5 ]{\file};
					\addplot[teal!80, thick, draw opacity = 0.5]      table[x = _0 , y  = _6 ]{\file};
					\addplot[violet!80, thick, draw opacity = 0.5]      table[x = _0 , y  = _7 ]{\file};
					\addplot[red!80, dashed, thick, draw opacity = 0.5]      table[x = _0 , y  = _8 ]{\file};
					\addplot[blue!80, dashed, thick, draw opacity = 0.5]      table[x = _0 , y  = _9 ]{\file};
					\addplot[green!80, dashed, thick, draw opacity = 0.5]      table[x = _0 , y  = _10 ]{\file};
					\addplot[cyan!80, dashed, thick, draw opacity = 0.5]      table[x = _0 , y  = _11 ]{\file};
					\addplot[pink!80, dashed, thick, draw opacity = 0.5]      table[x = _0 , y  = _12 ]{\file};
					\addplot[teal!80, dashed, thick, draw opacity = 0.5]      table[x = _0 , y  = _13 ]{\file};
					\addplot[violet!80, dashed, thick, draw opacity = 0.5]      table[x = _0 , y  = _14 ]{\file};
					
				\end{axis}
			\end{tikzpicture}
			\vspace{-0.3cm}
			\caption{
				$\epsilon = 0.93$
			}
		\end{subfigure}
		\begin{subfigure}{1\textwidth}
			\begin{tikzpicture}
				\def\file{pictures/ang_ep0.97_op1.txt}
				\begin{axis}[xlabel={$k$},ylabel={$x_i^1$},
					xmin=0.2, ymin = -0.1, xmax = 9.8,ymax=1.1,legend columns=14,
					width=0.35\textwidth,height=3.5cm,legend style={at={(3.42,1.15)},anchor=east}]
					
					\addplot[red!80, thick, draw opacity = 0.5]      table[x = _0 , y  = _1 ]{\file};
					\addplot[blue!80, thick, draw opacity = 0.5]      table[x = _0 , y  = _2 ]{\file};
					\addplot[green!80, thick, draw opacity = 0.5]      table[x = _0 , y  = _3 ]{\file};
					\addplot[cyan!80, thick, draw opacity = 0.5]      table[x = _0, y  = _4 ]{\file};
					\addplot[pink!80, thick, draw opacity = 0.5]      table[x = _0 , y  = _5 ]{\file};
					\addplot[teal!80, thick, draw opacity = 0.5]      table[x = _0 , y  = _6 ]{\file};
					\addplot[violet!80, thick, draw opacity = 0.5]      table[x = _0 , y  = _7 ]{\file};
					\addplot[red!80, dashed, thick, draw opacity = 0.5]      table[x = _0 , y  = _8 ]{\file};
					\addplot[blue!80, dashed, thick, draw opacity = 0.5]      table[x = _0 , y  = _9 ]{\file};
					\addplot[green!80, dashed, thick, draw opacity = 0.5]      table[x = _0 , y  = _10 ]{\file};
					\addplot[cyan!80, dashed, thick, draw opacity = 0.5]      table[x = _0 , y  = _11 ]{\file};
					\addplot[pink!80, dashed, thick, draw opacity = 0.5]      table[x = _0 , y  = _12 ]{\file};
					\addplot[teal!80, dashed, thick, draw opacity = 0.5]      table[x = _0 , y  = _13 ]{\file};
					\addplot[violet!80, dashed, thick, draw opacity = 0.5]      table[x = _0 , y  = _14 ]{\file};
					
					\legend{
						$\!\! \bm{x}_1 \!\!$, $\!\! \bm{x}_2 \!\!$, $\!\! \bm{x}_3 \!\!$, $\!\! \bm{x}_4 \!\!$, $\!\! \bm{x}_5 \!\!$, $\!\! \bm{x}_6 \!\!$, $\!\! \bm{x}_7 \!\!$, 
						$\!\! \bm{x}_8 \!\!$, $\!\! \bm{x}_9 \!\!$, $\!\! \bm{x}_{10} \!\!$, $\!\! \bm{x}_{11} \!\!$, $\!\! \bm{x}_{12} \!\!$, $\!\! \bm{x}_{13} \!\!$, $\!\! \bm{x}_{14} \!\!$}
				\end{axis}
				\def\file{pictures/ang_ep0.97_op2.txt}
				\begin{axis}[xlabel={$k$},ylabel={$x_i^2$},
					xmin=0.2, ymin = -0.1, xmax = 9.8,ymax=1.1,legend columns=7,
					width=0.35\textwidth,height=3.5cm,legend style={at={(1,1.25)},anchor=east},
					ylabel shift=-0.2cm,
					xshift=5.8cm]
					
					\addplot[red!80, thick, draw opacity = 0.5]      table[x = _0 , y  = _1 ]{\file};
					\addplot[blue!80, thick, draw opacity = 0.5]      table[x = _0 , y  = _2 ]{\file};
					\addplot[green!80, thick, draw opacity = 0.5]      table[x = _0 , y  = _3 ]{\file};
					\addplot[cyan!80, thick, draw opacity = 0.5]      table[x = _0, y  = _4 ]{\file};
					\addplot[pink!80, thick, draw opacity = 0.5]      table[x = _0 , y  = _5 ]{\file};
					\addplot[teal!80, thick, draw opacity = 0.5]      table[x = _0 , y  = _6 ]{\file};
					\addplot[violet!80, thick, draw opacity = 0.5]      table[x = _0 , y  = _7 ]{\file};
					\addplot[red!80, dashed, thick, draw opacity = 0.5]      table[x = _0 , y  = _8 ]{\file};
					\addplot[blue!80, dashed, thick, draw opacity = 0.5]      table[x = _0 , y  = _9 ]{\file};
					\addplot[green!80, dashed, thick, draw opacity = 0.5]      table[x = _0 , y  = _10 ]{\file};
					\addplot[cyan!80, dashed, thick, draw opacity = 0.5]      table[x = _0 , y  = _11 ]{\file};
					\addplot[pink!80, dashed, thick, draw opacity = 0.5]      table[x = _0 , y  = _12 ]{\file};
					\addplot[teal!80, dashed, thick, draw opacity = 0.5]      table[x = _0 , y  = _13 ]{\file};
					\addplot[violet!80, dashed, thick, draw opacity = 0.5]      table[x = _0 , y  = _14 ]{\file};
					
				\end{axis}
				\def\file{pictures/ang_ep0.97_op3.txt}
				\begin{axis}[xlabel={$k$},ylabel={$x_i^2$},
					xmin=0.2, ymin = -0.1, xmax = 9.8,ymax=1.1,legend columns=7,
					width=0.35\textwidth,height=3.5cm,legend style={at={(1,1.25)},anchor=east},
					ylabel shift=-0.2cm,
					xshift=11.6cm]
					
					\addplot[red!80, thick, draw opacity = 0.5]      table[x = _0 , y  = _1 ]{\file};
					\addplot[blue!80, thick, draw opacity = 0.5]      table[x = _0 , y  = _2 ]{\file};
					\addplot[green!80, thick, draw opacity = 0.5]      table[x = _0 , y  = _3 ]{\file};
					\addplot[cyan!80, thick, draw opacity = 0.5]      table[x = _0, y  = _4 ]{\file};
					\addplot[pink!80, thick, draw opacity = 0.5]      table[x = _0 , y  = _5 ]{\file};
					\addplot[teal!80, thick, draw opacity = 0.5]      table[x = _0 , y  = _6 ]{\file};
					\addplot[violet!80, thick, draw opacity = 0.5]      table[x = _0 , y  = _7 ]{\file};
					\addplot[red!80, dashed, thick, draw opacity = 0.5]      table[x = _0 , y  = _8 ]{\file};
					\addplot[blue!80, dashed, thick, draw opacity = 0.5]      table[x = _0 , y  = _9 ]{\file};
					\addplot[green!80, dashed, thick, draw opacity = 0.5]      table[x = _0 , y  = _10 ]{\file};
					\addplot[cyan!80, dashed, thick, draw opacity = 0.5]      table[x = _0 , y  = _11 ]{\file};
					\addplot[pink!80, dashed, thick, draw opacity = 0.5]      table[x = _0 , y  = _12 ]{\file};
					\addplot[teal!80, dashed, thick, draw opacity = 0.5]      table[x = _0 , y  = _13 ]{\file};
					\addplot[violet!80, dashed, thick, draw opacity = 0.5]      table[x = _0 , y  = _14 ]{\file};
					
				\end{axis}
			\end{tikzpicture}
			\vspace{-0.3cm}
			\caption{
				$\epsilon = 0.97$
			}
		\end{subfigure}
		\begin{subfigure}{1\textwidth}
			\begin{tikzpicture}
				\def\file{pictures/ang_ep0.99_op1.txt}
				\begin{axis}[xlabel={$k$},ylabel={$x_i^1$},
					xmin=0.2, ymin = -0.1, xmax = 9.8,ymax=1.1,legend columns=14,
					width=0.35\textwidth,height=3.5cm,legend style={at={(3.42,1.15)},anchor=east}]
					
					\addplot[red!80, thick, draw opacity = 0.5]      table[x = _0 , y  = _1 ]{\file};
					\addplot[blue!80, thick, draw opacity = 0.5]      table[x = _0 , y  = _2 ]{\file};
					\addplot[green!80, thick, draw opacity = 0.5]      table[x = _0 , y  = _3 ]{\file};
					\addplot[cyan!80, thick, draw opacity = 0.5]      table[x = _0, y  = _4 ]{\file};
					\addplot[pink!80, thick, draw opacity = 0.5]      table[x = _0 , y  = _5 ]{\file};
					\addplot[teal!80, thick, draw opacity = 0.5]      table[x = _0 , y  = _6 ]{\file};
					\addplot[violet!80, thick, draw opacity = 0.5]      table[x = _0 , y  = _7 ]{\file};
					\addplot[red!80, dashed, thick, draw opacity = 0.5]      table[x = _0 , y  = _8 ]{\file};
					\addplot[blue!80, dashed, thick, draw opacity = 0.5]      table[x = _0 , y  = _9 ]{\file};
					\addplot[green!80, dashed, thick, draw opacity = 0.5]      table[x = _0 , y  = _10 ]{\file};
					\addplot[cyan!80, dashed, thick, draw opacity = 0.5]      table[x = _0 , y  = _11 ]{\file};
					\addplot[pink!80, dashed, thick, draw opacity = 0.5]      table[x = _0 , y  = _12 ]{\file};
					\addplot[teal!80, dashed, thick, draw opacity = 0.5]      table[x = _0 , y  = _13 ]{\file};
					\addplot[violet!80, dashed, thick, draw opacity = 0.5]      table[x = _0 , y  = _14 ]{\file};
					
					\legend{
						$\!\! \bm{x}_1 \!\!$, $\!\! \bm{x}_2 \!\!$, $\!\! \bm{x}_3 \!\!$, $\!\! \bm{x}_4 \!\!$, $\!\! \bm{x}_5 \!\!$, $\!\! \bm{x}_6 \!\!$, $\!\! \bm{x}_7 \!\!$, 
						$\!\! \bm{x}_8 \!\!$, $\!\! \bm{x}_9 \!\!$, $\!\! \bm{x}_{10} \!\!$, $\!\! \bm{x}_{11} \!\!$, $\!\! \bm{x}_{12} \!\!$, $\!\! \bm{x}_{13} \!\!$, $\!\! \bm{x}_{14} \!\!$}
				\end{axis}
				\def\file{pictures/ang_ep0.99_op2.txt}
				\begin{axis}[xlabel={$k$},ylabel={$x_i^2$},
					xmin=0.2, ymin = -0.1, xmax = 9.8,ymax=1.1,legend columns=7,
					width=0.35\textwidth,height=3.5cm,legend style={at={(1,1.25)},anchor=east},
					ylabel shift=-0.2cm,
					xshift=5.8cm]
					
					\addplot[red!80, thick, draw opacity = 0.5]      table[x = _0 , y  = _1 ]{\file};
					\addplot[blue!80, thick, draw opacity = 0.5]      table[x = _0 , y  = _2 ]{\file};
					\addplot[green!80, thick, draw opacity = 0.5]      table[x = _0 , y  = _3 ]{\file};
					\addplot[cyan!80, thick, draw opacity = 0.5]      table[x = _0, y  = _4 ]{\file};
					\addplot[pink!80, thick, draw opacity = 0.5]      table[x = _0 , y  = _5 ]{\file};
					\addplot[teal!80, thick, draw opacity = 0.5]      table[x = _0 , y  = _6 ]{\file};
					\addplot[violet!80, thick, draw opacity = 0.5]      table[x = _0 , y  = _7 ]{\file};
					\addplot[red!80, dashed, thick, draw opacity = 0.5]      table[x = _0 , y  = _8 ]{\file};
					\addplot[blue!80, dashed, thick, draw opacity = 0.5]      table[x = _0 , y  = _9 ]{\file};
					\addplot[green!80, dashed, thick, draw opacity = 0.5]      table[x = _0 , y  = _10 ]{\file};
					\addplot[cyan!80, dashed, thick, draw opacity = 0.5]      table[x = _0 , y  = _11 ]{\file};
					\addplot[pink!80, dashed, thick, draw opacity = 0.5]      table[x = _0 , y  = _12 ]{\file};
					\addplot[teal!80, dashed, thick, draw opacity = 0.5]      table[x = _0 , y  = _13 ]{\file};
					\addplot[violet!80, dashed, thick, draw opacity = 0.5]      table[x = _0 , y  = _14 ]{\file};
					
				\end{axis}
				\def\file{pictures/ang_ep0.99_op3.txt}
				\begin{axis}[xlabel={$k$},ylabel={$x_i^2$},
					xmin=0.2, ymin = -0.1, xmax = 9.8,ymax=1.1,legend columns=7,
					width=0.35\textwidth,height=3.5cm,legend style={at={(1,1.25)},anchor=east},
					ylabel shift=-0.2cm,
					xshift=11.6cm]
					
					\addplot[red!80, thick, draw opacity = 0.5]      table[x = _0 , y  = _1 ]{\file};
					\addplot[blue!80, thick, draw opacity = 0.5]      table[x = _0 , y  = _2 ]{\file};
					\addplot[green!80, thick, draw opacity = 0.5]      table[x = _0 , y  = _3 ]{\file};
					\addplot[cyan!80, thick, draw opacity = 0.5]      table[x = _0, y  = _4 ]{\file};
					\addplot[pink!80, thick, draw opacity = 0.5]      table[x = _0 , y  = _5 ]{\file};
					\addplot[teal!80, thick, draw opacity = 0.5]      table[x = _0 , y  = _6 ]{\file};
					\addplot[violet!80, thick, draw opacity = 0.5]      table[x = _0 , y  = _7 ]{\file};
					\addplot[red!80, dashed, thick, draw opacity = 0.5]      table[x = _0 , y  = _8 ]{\file};
					\addplot[blue!80, dashed, thick, draw opacity = 0.5]      table[x = _0 , y  = _9 ]{\file};
					\addplot[green!80, dashed, thick, draw opacity = 0.5]      table[x = _0 , y  = _10 ]{\file};
					\addplot[cyan!80, dashed, thick, draw opacity = 0.5]      table[x = _0 , y  = _11 ]{\file};
					\addplot[pink!80, dashed, thick, draw opacity = 0.5]      table[x = _0 , y  = _12 ]{\file};
					\addplot[teal!80, dashed, thick, draw opacity = 0.5]      table[x = _0 , y  = _13 ]{\file};
					\addplot[violet!80, dashed, thick, draw opacity = 0.5]      table[x = _0 , y  = _14 ]{\file};
					
				\end{axis}
			\end{tikzpicture}
			\vspace{-0.3cm}
			\caption{
				$\epsilon = 0.99$
			}
		\end{subfigure}
		\vspace{-0.3cm}
		\caption{
			Angle-based method
		}
		\vspace{-0.3cm}
		\label{fig:angle}
	\end{figure*}
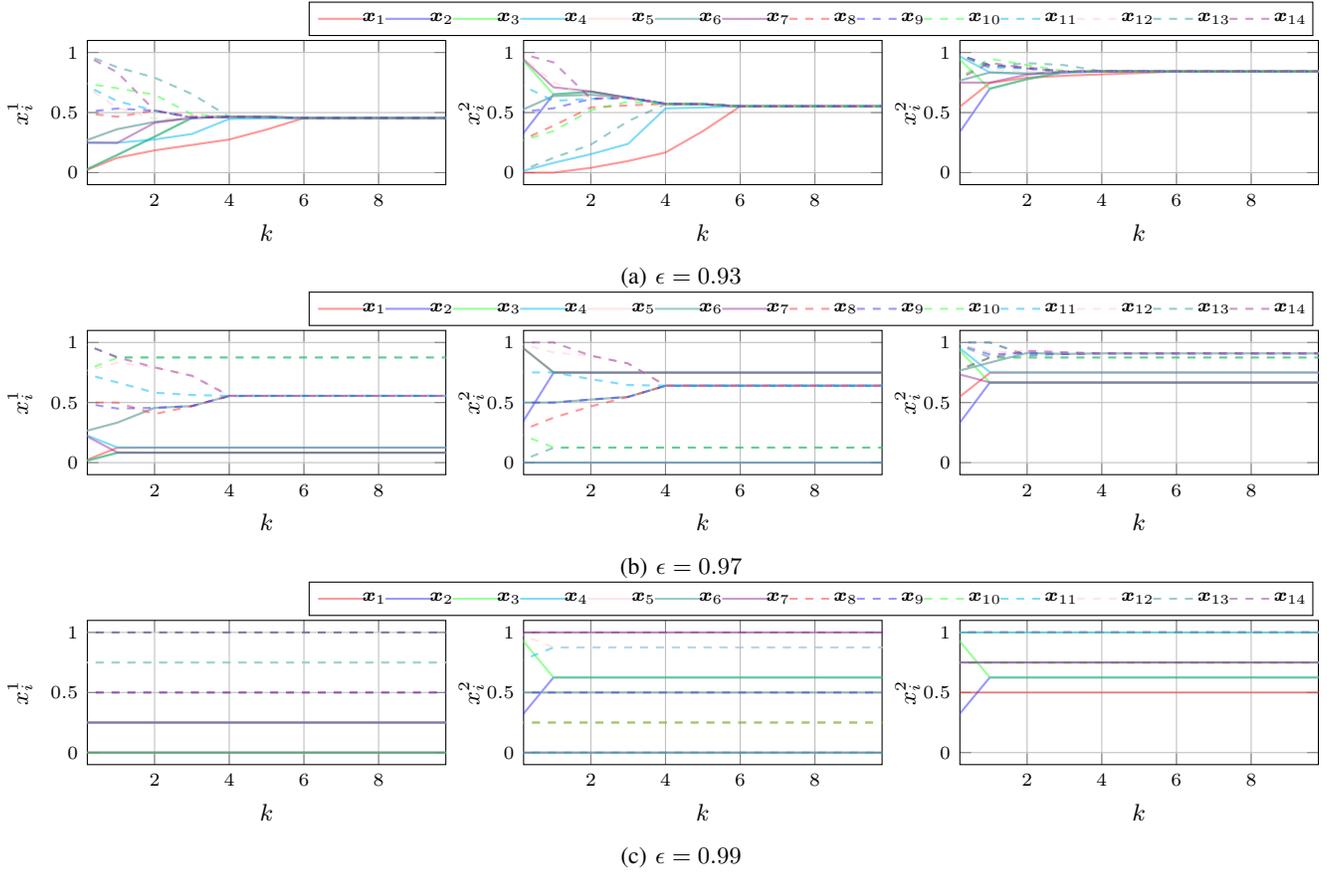
    
	\cref{figure_N_epsilon} illustrates the relationship between the angular threshold $\epsilon$ and the number of clusters. To cover more cases, the initial conditions are not taken from the Yelp dataset but are instead based on randomly generated initializations. 
	From \cref{figure_N_epsilon}, we observe that the number of clusters increases with $\epsilon$. Moreover, the figure shows that our proposed expression effectively captures the upper bound of the cluster number.
	
	\subsubsection{Discussion on $\epsilon$}
	
	Here, we further discuss the impact of $\epsilon$ on the number of clusters. Regarding sensitivity, as \cref{figure_N_epsilon_dis} and \cref{figure_N_epsilon} show for both distance- and angle-based similarities, the impact of small perturbations in $\epsilon$ is itself $\epsilon$-dependent. In particular, as $\epsilon \to 1$, a marginal increase in $\epsilon$ yields a disproportionately larger rise in the number of clusters.
	
	The impact of $\epsilon$ on the number of clusters also provides practical guidance for choosing these parameters in real applications.  
	In practice, $\epsilon$ serves as the threshold in the recommendation system, determining the minimum similarity required between two users to trigger mutual recommendations. A higher number of clusters indicates greater diversity in opinions, while fewer clusters reflect higher consensus. 
	In our results, a larger $\epsilon$ tends to split users into a greater number of smaller groups, leading to opinions being more widely dispersed across the spectrum. Conversely, a smaller $\epsilon$ tends to produce fewer but larger clusters, leading to a multipolar or even bipolar structure. The system achieves global consensus when $\epsilon$ falls below a certain threshold, which can be derived inversely using our proposed upper-bound. These results provide a basis for evaluating the societal impact of the parameters set by social media platforms.
	
	Our findings are consistent with prior empirical work. For example, the work in \cite{ramaciotti_morales_auditing_2021-1} examines similar mechanisms of recommendations and opinion dynamics from the perspective of auditing the effect, using Twitter data. They reach qualitatively comparable conclusions, namely that different recommendation principles sometimes either drive or mitigate polarization in real social networks. Similarly, \cite{bellina_effect_2023-1} introduces different parameters in their model, but with respect to the strength parameter of the similarity bias, they also conclude that when it is sufficiently large, the system tends to form polarized groups, thereby giving rise to the filter bubble effect.
	
	The experimental results validate the analysis that the parameter settings of recommendation systems significantly affect the final number of opinion clusters. This finding provides theoretical support for the idea that recommendation systems on social media platforms can lead to user fragmentation, potentially resulting in echo chambers and filter bubbles.
	
	\begin{figure}[t] 
		\centering
		\def\file{pictures/Group_epsilon.txt}
		\begin{tikzpicture}
			\begin{axis}[xlabel={$\epsilon$},ylabel={Nr. Groups},
				xmin=0.76, ymin = -0.1, xmax = 0.9999,ymax=131,legend columns=1,
				width=0.48\textwidth,height=3.5cm,legend pos= north west
				]
				
				\addplot[blue!80, only marks, mark=o, mark size=1]      table[x = Experiment_epsilon , y  = Experiment_GroupQuantity ]{\file};
				\addplot[red!80, only marks, mark=o, mark size=1]      table[x = Bound_epsilon , y  = Bound_GroupQuantity ]{\file};
				
				\legend{Experimental result, Theoretical upper bound}
			\end{axis}
		\end{tikzpicture}
		\caption{
			Number of clusters as $\epsilon$ varies in the proposed ODRS framework using the angle-based method.}
		\label{figure_N_epsilon}
	\end{figure}
	
	\subsection{Optimal Influence Strategy for Propagators}
	
	In this subsection, the effectiveness of the reinforcement learning-based control strategy is demonstrated by using the Yelp data set \cite{_yelp_}.
	To highlight the advantages of the selected Reinforcement Learning (RL) approach, we conduct a comparison with an alternative intelligent algorithm, namely Evolutionary Algorithms (EA). The following content first introduces the experimental setup, and then presents the results.
	
	We extract a subset from the Yelp dataset, which includes 8 users and opinions represented in 3 dimensions. To shape the users' opinion, $2$ propagators are introduced, whose opinions are generated via a reinforcement learning algorithm and not affected by the users.
	The target opinion $\bm{x}^c$ is randomly selected in the valid domain, i.e., $\bm{x}^c \in [0,1]^3$. For a clear illustration of the influence performance, an exemplar target is set as $\bm{x}^c = [0.8147, 0.9058, 0.1270]^T$.
	Moreover, the control horizon is set as $20$ steps.
	The threshold $\epsilon$ in \eqref{f_bar} and \eqref{g_bar} is chosen as $0.2$.
	
	The structures of the critic and actor networks are shown in \cref{fig_network_structure}.
	Specifically, the concatenated system state $\bm{X} \in \mathbb{R}^{8 \times 3}$ is reshaped into a $24$-dimensional vector, which combines with the target opinion $\bm{x}^c \in \mathbb{R}^3$ and serves as the $27$-dimensional input for critic and actor networks.
	The critic network, which generates the prediction of the cumulative reward with dimension $1$, has $1$ hidden layers, which applies a fully connected layer with $256$ hidden nodes and ReLU activation function.
	The actor network returns the concatenated mean prediction and its deviation, whose common layer consists of a fully connected layer with $256$ hidden nodes and is activated by the ReLU function.
	The $6$-dimensional output of the common layer is the result of a fully connected layer, which serves as the input for the mean and variance networks.
	Both network components for mean and variance are composed of a ReLU layer, followed by a fully connected layer with $6$-dimensional output.
	While the Softplus function is chosen for variance, the sigmoid layer is used for mean, such that the control command belongs to $(0,1)^6$.
	The generated control command will reshape to two $3$-dimensional vectors $\bm{u}_1$ and $\bm{u}_2$ in \eqref{discrete model}.
	
	\begin{figure}[t]
		\centering
		\includegraphics[width = 0.48\textwidth]{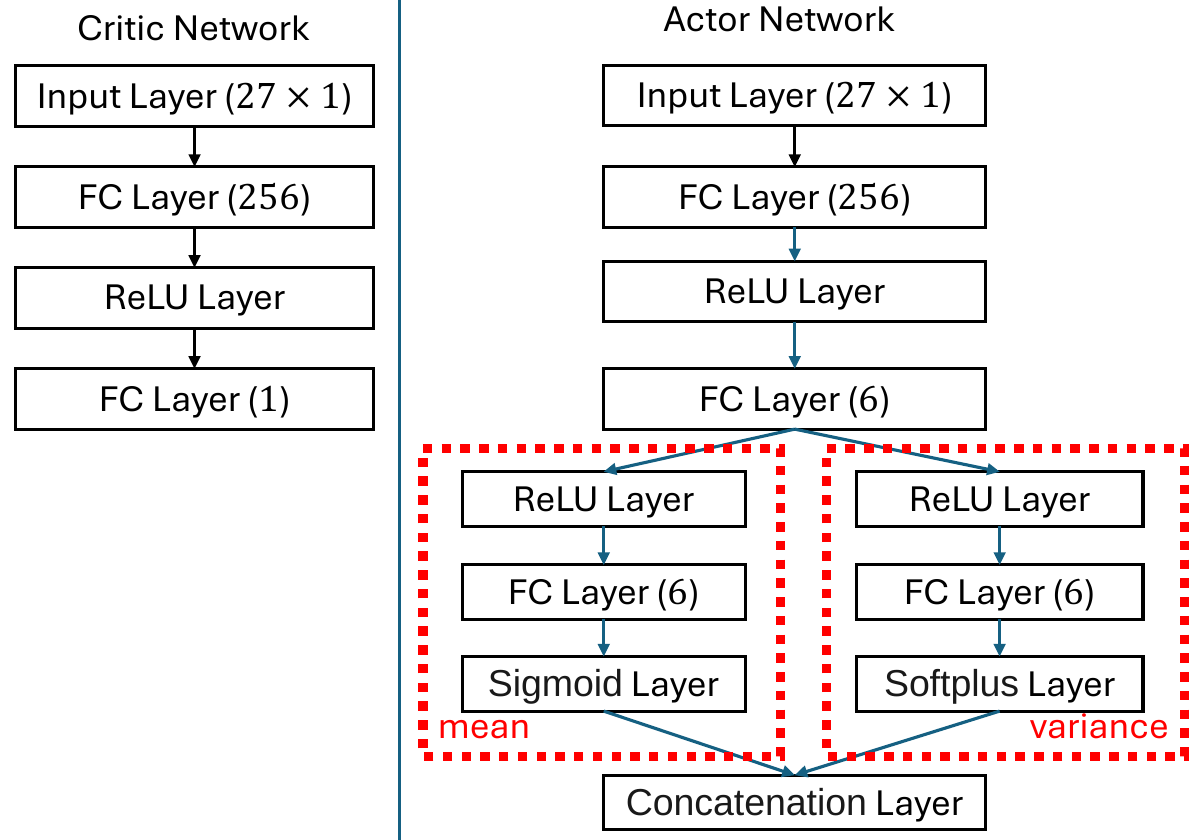}
		\vspace{-0.1cm}
		\caption{Structure of critic and actor networks in PPO reinforcement learning.}
		\vspace{-0.3cm}
		\label{fig_network_structure}
	\end{figure}
	
	The reinforcement learning-based agent is trained separately for distance- and angle-based methods, whose performance is shown as follows.
	
	\subsubsection{Distance-based Method}
	
	The reinforcement learning-based control policy is firstly trained in the proposed ODRS framework with distance-based method \eqref{s_f} and \eqref{f_bar}, where the training status is shown in \cref{figure_Training_Euclidean}.
	It is obvious to see the training converges by observing the episode reward, the average reward and the output of the critic network are almost identical after $8000$ episodes.
	The opinion evolution for users and propagators are shown in \cref{figure_State_Euclidean} and \cref{figure_Input_Euclidean} respectively, where the convergence of the user opinion to $\bm{x}^c$ is obvious.
	
	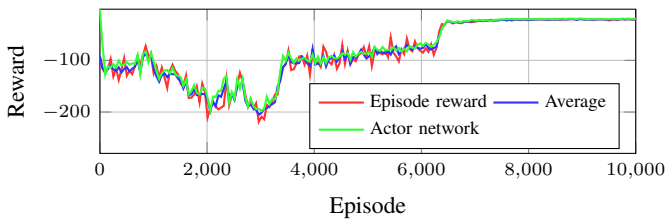
\begin{figure}[t] 
		\centering
		\def\file{pictures/Training_Euclidean.txt}
		\begin{tikzpicture}
			\begin{axis}[xlabel={Episode},ylabel={Reward},
				xmin=1, ymin = -280, xmax = 9999,ymax=-1,legend columns=2,
				width=0.48\textwidth,height=3.5cm,legend pos= south east
				]
				\addplot[red!80, thick]      table[x = EpisodeNrSet , y  = EpisodeReward ]{\file};
				\addplot[blue!80, thick]      table[x = EpisodeNrSet , y  = AverageReward ]{\file};
				\addplot[green!80, thick]      table[x = EpisodeNrSet , y  = EpisodeQ0 ]{\file};
				
				\legend{Episode reward, Average, Actor network}
			\end{axis}
		\end{tikzpicture}
		\caption{
			Training status for RL-based control strategy in the ODRS framework with the distance-based method, which is reflected by the cumulative rewards over episodes.
		}
		\label{figure_Training_Euclidean}
	\end{figure}
	
	\begin{figure}[t]
		\centering
		\begin{tikzpicture}
			
			\def\file{pictures/State_Euclidean.txt}
			\begin{axis}[xlabel={},ylabel={User $1$},
				xmin=0.1, ymin = -0.1, xmax = 19.9,ymax=1.65,legend columns=3,
				width=0.26\textwidth,height=3.5cm,legend pos=north east,
				xticklabels={,,},
				ylabel shift=-0.2cm]
				\def\file{pictures/State_Euclidean.txt}
				\addplot[red!80, thick]      table[x = k_set , y  = x11 ]{\file};
				\addplot[blue!80, thick]      table[x = k_set , y  = x12 ]{\file};
				\addplot[green!80, thick]      table[x = k_set , y  = x13 ]{\file};
				
				\def\file{pictures/State_EA_dis.txt}
				\addplot[red!80, thick, dashed]      table[x = k_set , y  = x11 ]{\file};
				\addplot[blue!80, thick, dashed]      table[x = k_set , y  = x12 ]{\file};
				\addplot[green!80, thick, dashed]      table[x = k_set , y  = x13 ]{\file};
				
				\def\file{pictures/xc.txt}
				\addplot[red!80, thick, dotted]      table[x = k_set_xc , y  = xc_1 ]{\file};
				\addplot[blue!80, thick, dotted]      table[x = k_set_xc , y  = xc_2 ]{\file};
				\addplot[green!80, thick, dotted]      table[x = k_set_xc , y  = xc_3 ]{\file};
				
				\legend{$\!\! x_1^1 \!\!$, $\!\! x_1^2 \!\!$, $\!\! x_1^3 \!\!$}
			\end{axis}
			\begin{axis}[xlabel={},ylabel={User $2$},
				xmin=0.1, ymin = -0.1, xmax = 19.9,ymax=1.65,legend columns=3,
				width=0.26\textwidth,height=3.5cm,legend pos=north east,
				xticklabels={,,},
				xshift=4cm,
				ylabel shift=-0.2cm]
				\def\file{pictures/State_Euclidean.txt}
				\addplot[red!80, thick]      table[x = k_set , y  = x21 ]{\file};
				\addplot[blue!80, thick]      table[x = k_set , y  = x22 ]{\file};
				\addplot[green!80, thick]      table[x = k_set , y  = x23 ]{\file};
				
				\def\file{pictures/State_EA_dis.txt}
				\addplot[red!80, thick, dashed]      table[x = k_set , y  = x21 ]{\file};
				\addplot[blue!80, thick, dashed]      table[x = k_set , y  = x22 ]{\file};
				\addplot[green!80, thick, dashed]      table[x = k_set , y  = x23 ]{\file};
				
				\def\file{pictures/xc.txt}
				\addplot[red!80, thick, dotted]      table[x = k_set_xc , y  = xc_1 ]{\file};
				\addplot[blue!80, thick, dotted]      table[x = k_set_xc , y  = xc_2 ]{\file};
				\addplot[green!80, thick, dotted]      table[x = k_set_xc , y  = xc_3 ]{\file};
				
				\legend{$\!\! x_2^1 \!\!$, $\!\! x_2^2 \!\!$, $\!\! x_2^3 \!\!$}
			\end{axis}
			\begin{axis}[xlabel={},ylabel={User $3$},
				xmin=0.1, ymin = -0.1, xmax = 19.9,ymax=1.65,legend columns=3,
				width=0.26\textwidth,height=3.5cm,legend pos=north east,
				xticklabels={,,},
				yshift=-2cm,
				ylabel shift=-0.2cm]
				\def\file{pictures/State_Euclidean.txt}
				\addplot[red!80, thick]      table[x = k_set , y  = x31 ]{\file};
				\addplot[blue!80, thick]      table[x = k_set , y  = x32 ]{\file};
				\addplot[green!80, thick]      table[x = k_set , y  = x33 ]{\file};
				
				\def\file{pictures/State_EA_dis.txt}
				\addplot[red!80, thick, dashed]      table[x = k_set , y  = x31 ]{\file};
				\addplot[blue!80, thick, dashed]      table[x = k_set , y  = x32 ]{\file};
				\addplot[green!80, thick, dashed]      table[x = k_set , y  = x33 ]{\file};
				
				\def\file{pictures/xc.txt}
				\addplot[red!80, thick, dotted]      table[x = k_set_xc , y  = xc_1 ]{\file};
				\addplot[blue!80, thick, dotted]      table[x = k_set_xc , y  = xc_2 ]{\file};
				\addplot[green!80, thick, dotted]      table[x = k_set_xc , y  = xc_3 ]{\file};
				
				\legend{$\!\! x_3^1 \!\!$, $\!\! x_3^2 \!\!$, $\!\! x_3^3 \!\!$}
			\end{axis}
			\begin{axis}[xlabel={},ylabel={User $4$},
				xmin=0.1, ymin = -0.1, xmax = 19.9,ymax=1.65,legend columns=3,
				width=0.26\textwidth,height=3.5cm,legend pos=north east,
				xticklabels={,,},
				yshift=-2cm,xshift=4cm,
				ylabel shift=-0.2cm]
				\def\file{pictures/State_Euclidean.txt}
				\addplot[red!80, thick]      table[x = k_set , y  = x41 ]{\file};
				\addplot[blue!80, thick]      table[x = k_set , y  = x42 ]{\file};
				\addplot[green!80, thick]      table[x = k_set , y  = x43 ]{\file};
				
				\def\file{pictures/State_EA_dis.txt}
				\addplot[red!80, thick, dashed]      table[x = k_set , y  = x41 ]{\file};
				\addplot[blue!80, thick, dashed]      table[x = k_set , y  = x42 ]{\file};
				\addplot[green!80, thick, dashed]      table[x = k_set , y  = x43 ]{\file};
				
				\def\file{pictures/xc.txt}
				\addplot[red!80, thick, dotted]      table[x = k_set_xc , y  = xc_1 ]{\file};
				\addplot[blue!80, thick, dotted]      table[x = k_set_xc , y  = xc_2 ]{\file};
				\addplot[green!80, thick, dotted]      table[x = k_set_xc , y  = xc_3 ]{\file};
				
				\legend{$\!\! x_4^1 \!\!$, $\!\! x_4^2 \!\!$, $\!\! x_4^3 \!\!$}
			\end{axis}
			\begin{axis}[xlabel={},ylabel={User $5$},
				xmin=0.1, ymin = -0.1, xmax = 19.9,ymax=1.65,legend columns=3,
				width=0.26\textwidth,height=3.5cm,legend pos=north east,
				xticklabels={,,},
				yshift=-4.0cm,
				ylabel shift=-0.2cm]
				\def\file{pictures/State_Euclidean.txt}
				\addplot[red!80, thick]      table[x = k_set , y  = x51 ]{\file};
				\addplot[blue!80, thick]      table[x = k_set , y  = x52 ]{\file};
				\addplot[green!80, thick]      table[x = k_set , y  = x53 ]{\file};
				
				\def\file{pictures/State_EA_dis.txt}
				\addplot[red!80, thick, dashed]      table[x = k_set , y  = x51 ]{\file};
				\addplot[blue!80, thick, dashed]      table[x = k_set , y  = x52 ]{\file};
				\addplot[green!80, thick, dashed]      table[x = k_set , y  = x53 ]{\file};
				
				\def\file{pictures/xc.txt}
				\addplot[red!80, thick, dotted]      table[x = k_set_xc , y  = xc_1 ]{\file};
				\addplot[blue!80, thick, dotted]      table[x = k_set_xc , y  = xc_2 ]{\file};
				\addplot[green!80, thick, dotted]      table[x = k_set_xc , y  = xc_3 ]{\file};
				
				\legend{$\!\! x_5^1 \!\!$, $\!\! x_5^2 \!\!$, $\!\! x_5^3 \!\!$}
			\end{axis}
			\begin{axis}[xlabel={},ylabel={User $6$},
				xmin=0.1, ymin = -0.1, xmax = 19.9,ymax=1.65,legend columns=3,
				width=0.26\textwidth,height=3.5cm,legend pos=north east,
				xticklabels={,,},
				yshift=-4.0cm,xshift=4cm,
				ylabel shift=-0.2cm]
				\def\file{pictures/State_Euclidean.txt}
				\addplot[red!80, thick]      table[x = k_set , y  = x61 ]{\file};
				\addplot[blue!80, thick]      table[x = k_set , y  = x62 ]{\file};
				\addplot[green!80, thick]      table[x = k_set , y  = x63 ]{\file};
				
				\def\file{pictures/State_EA_dis.txt}
				\addplot[red!80, thick, dashed]      table[x = k_set , y  = x61 ]{\file};
				\addplot[blue!80, thick, dashed]      table[x = k_set , y  = x62 ]{\file};
				\addplot[green!80, thick, dashed]      table[x = k_set , y  = x63 ]{\file};
				
				\def\file{pictures/xc.txt}
				\addplot[red!80, thick, dotted]      table[x = k_set_xc , y  = xc_1 ]{\file};
				\addplot[blue!80, thick, dotted]      table[x = k_set_xc , y  = xc_2 ]{\file};
				\addplot[green!80, thick, dotted]      table[x = k_set_xc , y  = xc_3 ]{\file};
				
				\legend{$\!\! x_6^1 \!\!$, $\!\! x_6^2 \!\!$, $\!\! x_6^3 \!\!$}
			\end{axis}
			\begin{axis}[xlabel={{\color{white} 1111111} $k$},ylabel={User $7$},
				xmin=0.1, ymin = -0.1, xmax = 19.9,ymax=1.65,legend columns=3,
				width=0.26\textwidth,height=3.5cm,legend pos=north east,
				xlabel style={text width=2.5cm},
				yshift=-6cm,
				ylabel shift=-0.2cm]
				\def\file{pictures/State_Euclidean.txt}
				\addplot[red!80, thick]      table[x = k_set , y  = x71 ]{\file};
				\addplot[blue!80, thick]      table[x = k_set , y  = x72 ]{\file};
				\addplot[green!80, thick]      table[x = k_set , y  = x73 ]{\file};
				
				\def\file{pictures/State_EA_dis.txt}
				\addplot[red!80, thick, dashed]      table[x = k_set , y  = x71 ]{\file};
				\addplot[blue!80, thick, dashed]      table[x = k_set , y  = x72 ]{\file};
				\addplot[green!80, thick, dashed]      table[x = k_set , y  = x73 ]{\file};
				
				\def\file{pictures/xc.txt}
				\addplot[red!80, thick, dotted]      table[x = k_set_xc , y  = xc_1 ]{\file};
				\addplot[blue!80, thick, dotted]      table[x = k_set_xc , y  = xc_2 ]{\file};
				\addplot[green!80, thick, dotted]      table[x = k_set_xc , y  = xc_3 ]{\file};
				
				\legend{$\!\! x_7^1 \!\!$, $\!\! x_7^2 \!\!$, $\!\! x_7^3 \!\!$}
			\end{axis}
			\begin{axis}[xlabel={{\color{white} 1111111} $k$},ylabel={User $8$},
				xmin=0.1, ymin = -0.1, xmax = 19.9,ymax=1.65,legend columns=3,
				width=0.26\textwidth,height=3.5cm,legend pos=north east,
				xlabel style={text width=2.5cm},
				yshift=-6cm,xshift=4cm,
				ylabel shift=-0.2cm]
				\def\file{pictures/State_Euclidean.txt}
				\addplot[red!80, thick]      table[x = k_set , y  = x81 ]{\file};
				\addplot[blue!80, thick]      table[x = k_set , y  = x82 ]{\file};
				\addplot[green!80, thick]      table[x = k_set , y  = x83 ]{\file};
				
				\def\file{pictures/State_EA_dis.txt}
				\addplot[red!80, thick, dashed]      table[x = k_set , y  = x81 ]{\file};
				\addplot[blue!80, thick, dashed]      table[x = k_set , y  = x82 ]{\file};
				\addplot[green!80, thick, dashed]      table[x = k_set , y  = x83 ]{\file};
				
				\def\file{pictures/xc.txt}
				\addplot[red!80, thick, dotted]      table[x = k_set_xc , y  = xc_1 ]{\file};
				\addplot[blue!80, thick, dotted]      table[x = k_set_xc , y  = xc_2 ]{\file};
				\addplot[green!80, thick, dotted]      table[x = k_set_xc , y  = xc_3 ]{\file};
				
				\legend{$\!\! x_8^1 \!\!$, $\!\! x_8^2 \!\!$, $\!\! x_8^3 \!\!$}
			\end{axis}
		\end{tikzpicture}
		\caption{
			The evolution of user opinion in ODRS framework with distance-based method. The dotted line is the target opinion, the solid line shows the RL result, and the dashed line shows the EA result.
		}
		\label{figure_State_Euclidean}
	\end{figure}
	
	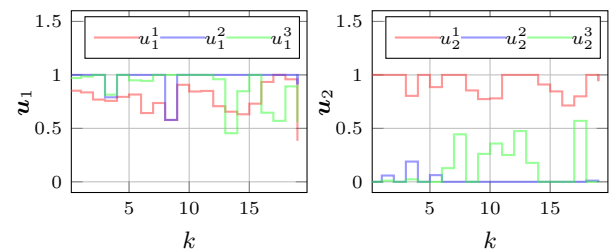
\begin{figure}[t] 
		\centering
		\def\file{pictures/Input_Euclidean.txt}
		\begin{tikzpicture}
			\begin{axis}[xlabel={$k$},ylabel={$\bm{u}_1$},
				xmin=0.2, ymin = -0.1, xmax = 19.8,ymax=1.6,legend columns=3,
				width=0.26\textwidth,height=4cm,legend pos=north east,
				ylabel shift=-0.2cm]
				\addplot[red!80, thick, draw opacity = 0.5]      table[x = k_set , y  = u11 ]{\file};
				\addplot[blue!80, thick, draw opacity = 0.5]      table[x = k_set , y  = u12 ]{\file};
				\addplot[green!80, thick, draw opacity = 0.5]      table[x = k_set , y  = u13 ]{\file};
				
				\legend{$\!\! u_1^1 \!\!$, $\!\! u_1^2 \!\!$, $\!\! u_1^3 \!\!$}
			\end{axis}
			\begin{axis}[xlabel={$k$},ylabel={$\bm{u}_2$},
				xmin=0.2, ymin = -0.1, xmax = 19.8,ymax=1.6,legend columns=3,
				width=0.26\textwidth,height=4cm,legend pos=north east,
				ylabel shift=-0.2cm,
				xshift=4cm]
				\addplot[red!80, thick, draw opacity = 0.5]      table[x = k_set , y  = u21 ]{\file};
				\addplot[blue!80, thick, draw opacity = 0.5]      table[x = k_set , y  = u22 ]{\file};
				\addplot[green!80, thick, draw opacity = 0.5]      table[x = k_set , y  = u23 ]{\file};
				
				\legend{$\!\! u_2^1 \!\!$, $\!\! u_2^2 \!\!$, $\!\! u_2^3 \!\!$}
			\end{axis}
		\end{tikzpicture}
		\caption{
			The evolution of propagator opinion in ODRS framework with distance-based method.
		}
		\label{figure_Input_Euclidean}
	\end{figure}
	
	\subsubsection{Angular Similarity}
	
	The training result for the ODRS framework with an angle-based method is shown in \cref{figure_Training_Angle}, demonstrating the convergence of the training process after $2000$ episodes.
	The opinion evolution for users and propagators is shown in \cref{figure_State_Angle} and \cref{figure_Input_Angle}, where the user opinions achieve consensus to $\bm{x}^c$.
	This demonstrates the effectiveness of the reinforcement learning-based control algorithm.
	
	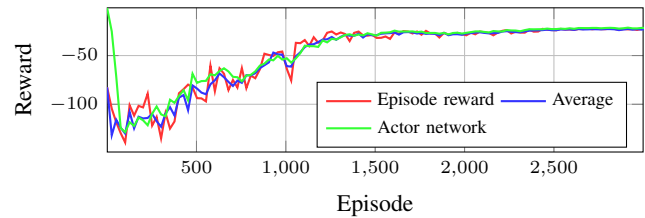
\begin{figure}[t] 
		\centering
		\def\file{pictures/Training_Angle.txt}
		\begin{tikzpicture}
			\begin{axis}[xlabel={Episode},ylabel={Reward},
				xmin=1, ymin = -149, xmax = 2999,ymax=-1,legend columns=2,
				width=0.48\textwidth,height=3.5cm,legend pos= south east
				]
				\addplot[red!80, thick]      table[x = EpisodeNrSet , y  = EpisodeReward ]{\file};
				\addplot[blue!80, thick]      table[x = EpisodeNrSet , y  = AverageReward ]{\file};
				\addplot[green!80, thick]      table[x = EpisodeNrSet , y  = EpisodeQ0 ]{\file};
				
				\legend{Episode reward, Average, Actor network}
			\end{axis}
		\end{tikzpicture}
		\caption{
			Training status for RL-based control strategy in the ODRS framework with the angle-based method, which is reflected by the cumulative rewards over episodes.
		}
		\label{figure_Training_Angle}
	\end{figure}
	
	\begin{figure}[t]
		\centering
		\begin{tikzpicture}
			\def\file{pictures/State_Angle.txt}
			
			\begin{axis}[xlabel={},ylabel={$x_i^1$},
				xmin=0.2, ymin = -0.1, xmax = 19.8,ymax=1.1,legend columns=5,
				width=0.48\textwidth,height=3.5cm,legend pos=south east,
				xticklabels={,,}]
				\def\file{pictures/State_Angle.txt}
				\addplot[red!80, thick, draw opacity = 0.5]      table[x = k_set , y  = x11 ]{\file};
				\addplot[blue!80, thick, draw opacity = 0.5]      table[x = k_set , y  = x21 ]{\file};
				\addplot[green!80, thick, draw opacity = 0.5]      table[x = k_set , y  = x31 ]{\file};
				\addplot[cyan!80, thick, draw opacity = 0.5]      table[x = k_set , y  = x41 ]{\file};
				\addplot[pink!80, thick, draw opacity = 0.5]      table[x = k_set , y  = x51 ]{\file};
				\addplot[teal!80, thick, draw opacity = 0.5]      table[x = k_set , y  = x61 ]{\file};
				\addplot[violet!80, thick, draw opacity = 0.5]      table[x = k_set , y  = x71 ]{\file};
				\addplot[yellow!80, thick, draw opacity = 0.5]      table[x = k_set , y  = x81 ]{\file};
				
				\def\file{pictures/State_EA.txt}
				\addplot[red!80, thick, draw opacity = 0.5, dashed]      table[x = k_set , y  = x11 ]{\file};
				\addplot[blue!80, thick, draw opacity = 0.5, dashed]      table[x = k_set , y  = x21 ]{\file};
				\addplot[green!80, thick, draw opacity = 0.5, dashed]      table[x = k_set , y  = x31 ]{\file};
				\addplot[cyan!80, thick, draw opacity = 0.5, dashed]      table[x = k_set , y  = x41 ]{\file};
				\addplot[pink!80, thick, draw opacity = 0.5, dashed]      table[x = k_set , y  = x51 ]{\file};
				\addplot[teal!80, thick, draw opacity = 0.5, dashed]      table[x = k_set , y  = x61 ]{\file};
				\addplot[violet!80, thick, draw opacity = 0.5, dashed]      table[x = k_set , y  = x71 ]{\file};
				\addplot[yellow!80, thick, draw opacity = 0.5, dashed]      table[x = k_set , y  = x81 ]{\file};
				
				\def\file{pictures/xc.txt}
				\addplot[black!80, thick, draw opacity = 0.5, dotted]      table[x = k_set_xc , y  = xc_2 ]{\file};
				\legend{$\! x_1^1 \!$, $\! x_2^1 \!$, $\! x_3^1 \!$, $\! x_4^1 \!$, $\! x_5^1 \!$, $\! x_6^1 \!$, $\! x_7^1 \!$, $\! x_8^1 \!$}
			\end{axis}
			
			\begin{axis}[xlabel={},ylabel={$x_i^2$},
				xmin=0.2, ymin = -0.1, xmax = 19.8,ymax=1.1,legend columns=5,
				width=0.48\textwidth,height=3.5cm,legend pos=south east,
				xticklabels={,,},
				yshift=-2.0cm]
				\def\file{pictures/State_Angle.txt}
				\addplot[red!80, thick, draw opacity = 0.5]      table[x = k_set , y  = x12 ]{\file};
				\addplot[blue!80, thick, draw opacity = 0.5]      table[x = k_set , y  = x22 ]{\file};
				\addplot[green!80, thick, draw opacity = 0.5]      table[x = k_set , y  = x32 ]{\file};
				\addplot[cyan!80, thick, draw opacity = 0.5]      table[x = k_set , y  = x42 ]{\file};
				\addplot[pink!80, thick, draw opacity = 0.5]      table[x = k_set , y  = x52 ]{\file};
				\addplot[teal!80, thick, draw opacity = 0.5]      table[x = k_set , y  = x62 ]{\file};
				\addplot[violet!80, thick, draw opacity = 0.5]      table[x = k_set , y  = x72 ]{\file};
				\addplot[yellow!80, thick, draw opacity = 0.5]      table[x = k_set , y  = x82 ]{\file};
				
				\def\file{pictures/State_EA.txt}
				\addplot[red!80, thick, draw opacity = 0.5, dashed]      table[x = k_set , y  = x12 ]{\file};
				\addplot[blue!80, thick, draw opacity = 0.5, dashed]      table[x = k_set , y  = x22 ]{\file};
				\addplot[green!80, thick, draw opacity = 0.5, dashed]      table[x = k_set , y  = x32 ]{\file};
				\addplot[cyan!80, thick, draw opacity = 0.5, dashed]      table[x = k_set , y  = x42 ]{\file};
				\addplot[pink!80, thick, draw opacity = 0.5, dashed]      table[x = k_set , y  = x52 ]{\file};
				\addplot[teal!80, thick, draw opacity = 0.5, dashed]      table[x = k_set , y  = x62 ]{\file};
				\addplot[violet!80, thick, draw opacity = 0.5, dashed]      table[x = k_set , y  = x72 ]{\file};
				\addplot[yellow!80, thick, draw opacity = 0.5, dashed]      table[x = k_set , y  = x82 ]{\file};
				
				\def\file{pictures/xc.txt}
				\addplot[black!80, thick, draw opacity = 0.5, dotted]      table[x = k_set_xc , y  = xc_1 ]{\file};
				
				\legend{$\! x_1^2 \!$, $\! x_2^2 \!$, $\! x_3^2 \!$, $\! x_4^2 \!$, $\! x_5^2 \!$, $\! x_6^2 \!$, $\! x_7^2 \!$, $\! x_8^2 \!$}
			\end{axis}
			\begin{axis}[xlabel={$k$},ylabel={$x_i^3$},
				xmin=0.2, ymin = -0.1, xmax = 19.8,ymax=1.1,legend columns=5,
				width=0.48\textwidth,height=3.5cm,legend pos=north east,
				yshift=-4.0cm]
				\def\file{pictures/State_Angle.txt}
				\addplot[red!80, thick, draw opacity = 0.5]      table[x = k_set , y  = x13 ]{\file};
				\addplot[blue!80, thick, draw opacity = 0.5]      table[x = k_set , y  = x23 ]{\file};
				\addplot[green!80, thick, draw opacity = 0.5]      table[x = k_set , y  = x33 ]{\file};
				\addplot[cyan!80, thick, draw opacity = 0.5]      table[x = k_set , y  = x43 ]{\file};
				\addplot[pink!80, thick, draw opacity = 0.5]      table[x = k_set , y  = x53 ]{\file};
				\addplot[teal!80, thick, draw opacity = 0.5]      table[x = k_set , y  = x63 ]{\file};
				\addplot[violet!80, thick, draw opacity = 0.5]      table[x = k_set , y  = x73 ]{\file};
				\addplot[yellow!80, thick, draw opacity = 0.5]      table[x = k_set , y  = x83 ]{\file};
				
				\def\file{pictures/State_EA.txt}
				\addplot[red!80, thick, draw opacity = 0.5, dashed]      table[x = k_set , y  = x13 ]{\file};
				\addplot[blue!80, thick, draw opacity = 0.5, dashed]      table[x = k_set , y  = x23 ]{\file};
				\addplot[green!80, thick, draw opacity = 0.5, dashed]      table[x = k_set , y  = x33 ]{\file};
				\addplot[cyan!80, thick, draw opacity = 0.5, dashed]      table[x = k_set , y  = x43 ]{\file};
				\addplot[pink!80, thick, draw opacity = 0.5, dashed]      table[x = k_set , y  = x53 ]{\file};
				\addplot[teal!80, thick, draw opacity = 0.5, dashed]      table[x = k_set , y  = x63 ]{\file};
				\addplot[violet!80, thick, draw opacity = 0.5, dashed]      table[x = k_set , y  = x73 ]{\file};
				\addplot[yellow!80, thick, draw opacity = 0.5, dashed]      table[x = k_set , y  = x83 ]{\file};
				
				\def\file{pictures/xc.txt}
				\addplot[black!80, thick, draw opacity = 0.5, dotted]      table[x = k_set_xc , y  = xc_3 ]{\file};
				
				\legend{$\! x_1^3 \!$, $\! x_2^3 \!$, $\! x_3^3 \!$, $\! x_4^3 \!$, $\! x_5^3 \!$, $\! x_6^3 \!$, $\! x_7^3 \!$, $\! x_8^3 \!$}
			\end{axis}
		\end{tikzpicture}
		\caption{
			The evolution of user opinion in ODRS framework with angle-based method. The dotted line is the target opinion, the solid line shows the RL result, and the dashed line shows the EA result.
		}
		\label{figure_State_Angle}
	\end{figure}
	
	\begin{figure}[t] 
		\centering
		\def\file{pictures/Input_Angle.txt}
		\begin{tikzpicture}
			\begin{axis}[xlabel={$k$},ylabel={$\bm{u}_1$},
				xmin=0.2, ymin = -0.1, xmax = 19.8,ymax=1.6,legend columns=3,
				width=0.26\textwidth,height=4cm,legend pos=north east,
				ylabel shift=-0.2cm]
				\addplot[red!80, thick, draw opacity = 0.5]      table[x = k_set , y  = u11 ]{\file};
				\addplot[blue!80, thick, draw opacity = 0.5]      table[x = k_set , y  = u12 ]{\file};
				\addplot[green!80, thick, draw opacity = 0.5]      table[x = k_set , y  = u13 ]{\file};
				
				\legend{$\!\! u_1^1 \!\!$, $\!\! u_1^2 \!\!$, $\!\! u_1^3 \!\!$}
			\end{axis}
			\begin{axis}[xlabel={$k$},ylabel={$\bm{u}_2$},
				xmin=0.2, ymin = -0.1, xmax = 19.8,ymax=1.6,legend columns=3,
				width=0.26\textwidth,height=4cm,legend pos=north east,
				ylabel shift=-0.2cm,
				xshift=4cm]
				\addplot[red!80, thick, draw opacity = 0.5]      table[x = k_set , y  = u21 ]{\file};
				\addplot[blue!80, thick, draw opacity = 0.5]      table[x = k_set , y  = u22 ]{\file};
				\addplot[green!80, thick, draw opacity = 0.5]      table[x = k_set , y  = u23 ]{\file};
				
				\legend{$\!\! u_2^1 \!\!$, $\!\! u_2^2 \!\!$, $\!\! u_2^3 \!\!$}
			\end{axis}
		\end{tikzpicture}
		\vspace{-0.3cm}
		\caption{
			The evolution of propagator opinion in ODRS framework with angle-based method.
		}
		\vspace{-0.3cm}
		\label{figure_Input_Angle}
	\end{figure}
	
	\subsubsection{Evolutionary Algorithm}

    For comparison, we develop a simple evolutionary algorithm to guide the propagators without using learning techniques.
	Specifically, we represent the control inputs over the entire time horizon as individuals in the population, and use the cost function as the fitness function. The next generation is produced using a combination of Gaussian mutation and uniform crossover. A tournament selection is employed to select elites as parents, which are then used to generate new offspring. This iterative process continues until the fitness does not improve anymore.
	
	The comparison results are given in \cref{figure_State_Euclidean} for the distance-based method and \cref{figure_State_Angle} for the angle-based method, where different line styles are used to represent the results of EA and RL.
	As the results in the figures are not sufficiently clear, we summarize the performance of the two algorithms in \cref{table_drl_ea_comparison}. As shown, both algorithms achieve comparable final costs, with RL performing slightly better than EA. In terms of tracking the target opinion, RL yields a lower average error compared to EA. However, EA converges with significantly fewer iterations than RL. Based on these observations, we conclude that while RL may produce higher-quality solutions in this scenario, EA requires less computational effort.
	\begin{table}[!t]
		\renewcommand{\arraystretch}{1.3}
		\caption{Performance Comparison Between RL and EA}
		\label{table_drl_ea_comparison}
		\centering
		\begin{tabular}{c c c c c c c}
			\hline
			\textbf{Metrics} & \textbf{Method} & \textbf{Final Cost} & \textbf{Avg Deviation} &  \textbf{Iteration Times}\\
			\hline
			Distance & RL & \textbf{203.4198} & \textbf{0.0942} & 3000\\
			Distance & EA  & 241.4343 & 0.1434 & \textbf{272} \\
			\hline
			Angle & RL & \textbf{232.7229} & \textbf{0.1099} & 3000\\
			Angle & EA  & 268.6896 & 0.1568 & \textbf{276} \\
			\hline
		\end{tabular}
	\end{table}

	\section{Conclusion}
	\label{section_conclusion}
	
	In this paper, a framework combining opinion dynamics and a recommendation system is proposed, i.e., the ODRS framework, where the connection between individuals follows either a distance- or an angle-based method.
	The proposed framework is analytical, where the convergence and clustering properties are formally proven.
	Moreover, the ODRS framework allows the introduction of a controller, i.e., propagators, whose behavior is modeled as an optimal control problem to shape users' opinions.
	Due to the high nonlinearity of the proposed framework, the optimal propagation strategy is obtained using reinforcement learning.
	The properties of the ODRS framework and the effectiveness of the reinforcement learning-based opinion-shaping strategy are demonstrated via simulations with the actual Yelp data set.
	
	We believe that our work contributes to a better quantitative understanding of the societal impacts of recommendation algorithms, supports the fair application of such systems, and offers insights toward the broader goal of promoting responsible and ethical technology design.
	
	\bibliographystyle{IEEEtran}
	\bibliography{refs}
	
	\vspace{-33pt}
	\begin{IEEEbiography}[{\includegraphics[width=1in,height=1.25in,clip,keepaspectratio]{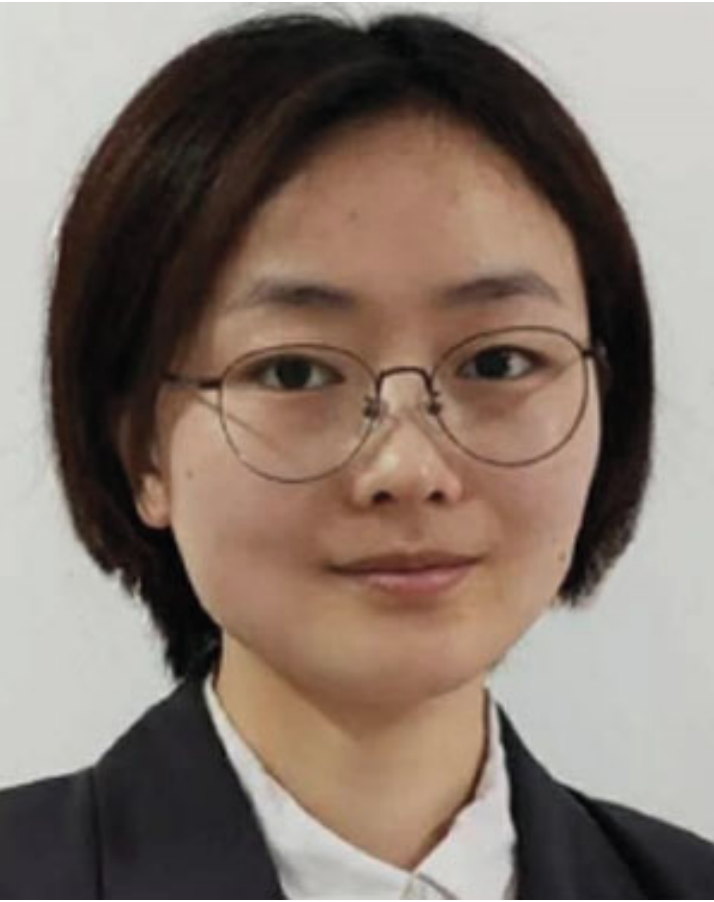}}]{Yuhong Chen}
		received the B.Eng. and M.Sc. degrees in control theory and engineering from the Harbin Institute of Technology, Harbin, China, in 2017 and 2019, respectively. She is currently working toward the Ph.D. degree with the Chair of Automatic Control Engineering, Technical University of Munich, Munich, Germany.
		
		Her current research interests include modeling, analysis, and control on social networks.
	\end{IEEEbiography}
	\vspace{-30pt}
	
	\begin{IEEEbiography}[{\includegraphics[width=1in,height=1.25in,clip,keepaspectratio]{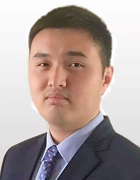}}]{Xiaobing Dai}
		received the B.Sc. mechanical engineering from the Tongji University, Shanghai, China, in 2018 with direction in mechatronics, building environment and civil engineering. He received double M.Sc degrees in Mechanical Engineering, Mechatronics and Robotics from the Technical University of Munich, Munich, Germany, in 2021. 
		Since February 2022, he is a PhD student at the Chair of Information-oriented Control, TUM School of Computation, Information and Technology at the Technical University of Munich, Munich, Germany. 
		
		His current research interests include efficient machine learning, networked control systems, and safe learning-based control.
	\end{IEEEbiography}
	\vspace{-30pt}
	
	\begin{IEEEbiography}[{\includegraphics[width=1in,height=1.25in,clip,keepaspectratio]{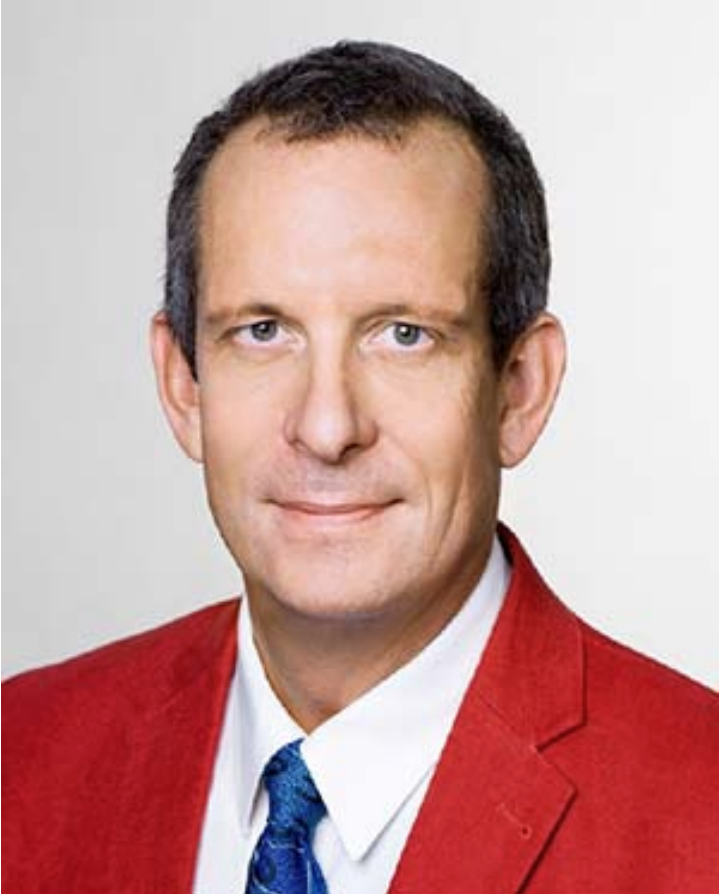}}]{Martin Buss}
		(Fellow, IEEE)
		received the diploma engineering degree in electrical engineering from the Technische Universität Darmstadt, Darmstadt, Germany, in 1990, and the doctor of engineering degree in electrical engineering from The University of Tokyo, Tokyo, Japan, in 1994.
		
		Since 2003, he has been a Full Professor (Chair) with the Chair of Automatic Control Engineering, Faculty of Electrical Engineering and Information Technology, Technical University of Munich. His research interests include automatic control, mechatronics, multimodal human system interfaces, optimization, nonlinear, and hybrid discrete-continuous systems.
		
		Dr. Buss was the recipient of the ERC Advanced Grant SHRINE.
	\end{IEEEbiography}
	\vspace{-30pt}
	
	\begin{IEEEbiography}[{\includegraphics[width=1in,height=1.25in,clip,keepaspectratio]{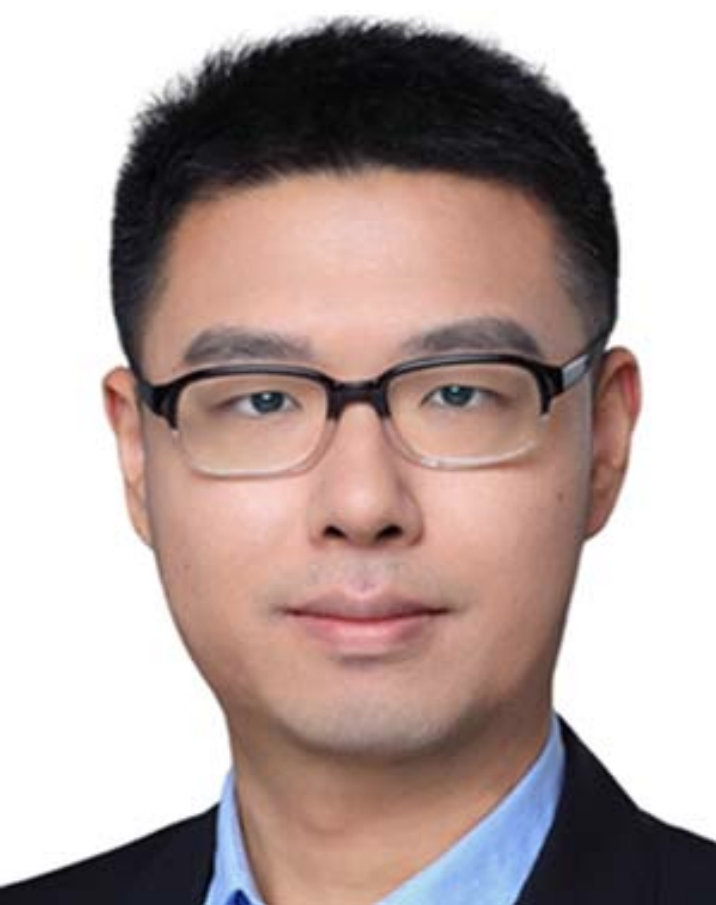}}]{Fangzhou Liu}
		(Member, IEEE)
		received the Doktor-Ingenieur degree in electrical engineering from the Technical University of Munich, Munich, Germany, in 2019.
		
		He was a Lecturer and a Research Fellow with the Chair of Automatic Control Engineering, Technical University of Munich. He is currently a Full Professor with the School of Astronautics, Harbin Institute of Technology, Harbin, China. His current research interests include networked control systems, reinforcement learning, and their applications.
	\end{IEEEbiography}
	\vfill
	
\end{document}